\documentclass[reqno,xcolor=dvipsnames]{amsart}
\usepackage[dvipsnames]{xcolor}

\usepackage{geometry}
\usepackage[round,comma]{natbib}                

\usepackage{enumerate}
\usepackage{mathabx}

\usepackage{amsmath}
\usepackage{amsthm}
\usepackage{amsfonts}
\usepackage{amssymb}
\usepackage{dsfont}

\usepackage{nicefrac}
\usepackage{booktabs}
\usepackage{mathrsfs}
\usepackage{bm} 
\usepackage{mathtools}                          

\newcommand{\ccB}{{\mathscr B}}

\newcommand{\ccF}{{\mathscr F}}\newcommand{\cF}{{\mathcal F}}

\newcommand{\ccP}{{\mathscr P}}
\newcommand{\ccE}{{\mathscr E}}

\newcommand{\ccN}{{\mathscr N}}\newcommand{\cN}{{\mathcal N}}
\newcommand{\cS}{{\mathcal S}}

\newcommand{\cT}{{\mathcal T}}

\newcommand{\im}{{\rm i}}

\DeclareMathOperator{\Var}{Var}
\DeclareMathOperator{\Cov}{Cov}

\newcommand{\Ind}{{\mathds 1}}
\newcommand{\ind}[1]{\Ind_{\{#1\}}}

\newcommand{\R}{\mathbb{R}}

\newcommand{\FF}{\mathbb{F}}

\newcommand{\N}{\mathbb{N}}

\newcommand{\C}{\mathbb{C}}

\renewcommand{\cF}{\ccF}

\newcommand{\Lloc}{L^2_{\rm loc}}

\newtheorem{theorem}{Theorem}[section]
\newtheorem{corollary}[theorem]{Corollary}      
\newtheorem{lemma}[theorem]{Lemma}              
\newtheorem{proposition}[theorem]{Proposition}  

\theoremstyle{definition}
\newtheorem{example}[theorem]{Example} 
\newtheorem{definition}[theorem]{Definition} 
\newtheorem{remark}[theorem]{Remark}
\newtheorem{assumption}[theorem]{Assumption}

\numberwithin{equation}{section}

\usepackage[backgroundcolor=white,bordercolor=orange]{todonotes}

\usepackage{changes}      
\definechangesauthor[name = Thorsten Schmidt, color=PineGreen]{TS}

\usepackage[utf8x]{inputenc}
\usepackage{lmodern,textcomp}

\renewcommand{\P}{P}

\newenvironment{prooflemma}[1]{{\parindent 0pt \it Proof of Lemma #1.}}{\mbox{}\hfill\mbox{$\Box\hspace{-0.5mm}$}\vskip 16pt}

\newcommand{\dbra}[1]{[\kern-0.15em[ #1 ]\kern-0.15em]}
\newcommand{\dbraco}[1]{[\kern-0.15em[ #1 [\kern-0.15em[}

\newcommand{\dinttT}{\int_{(t,T]^2}}
\usepackage[colorlinks,urlcolor=black,citecolor=blue,linkcolor=red]{hyperref}

\begin{document}

\title{Term structure modelling with overnight rates beyond stochastic continuity}

	\author{Claudio Fontana}
	\address{University of Padova, Department of Mathematics ``Tullio Levi-Civita'', via Trieste 63, 35121 Padova, Italy.}
	\email{fontana@math.unipd.it}
		\author{Zorana Grbac}
		\address{Paris Cit\'{e} University, Laboratoire de Probabilit\'es, Statistique et Mod\'elisation, 
		8 Pl. Aur\'elie Nemours, 75013 Paris, France.}
    \email{grbac@lpsm.paris}
		\author{Thorsten Schmidt}
		\address{Albert-Ludwigs University of Freiburg, Ernst-Zermelo Str. 1, 79104 Freiburg, Germany.}
    \email{ts@stochastik.uni-freiburg.de}
    \thanks{
    \indent The authors are thankful to two anonymous reviewers for constructive comments that helped to improve the paper. Financial support from the University of Padova (research programme BIRD190200/19, ``Term Structure Dynamics in Interest Rate and Energy Markets'') and the Europlace Institute of Finance is gratefully acknowledged.
    Data sharing not applicable to this article as no datasets were generated or analysed during the current study.}
    \date{\today. }

\maketitle

\begin{abstract}
Overnight rates, such as the SOFR (Secured Overnight Financing Rate) in the US, are central to the current reform of interest rate benchmarks.
A striking feature of overnight rates is the presence of jumps and spikes occurring at predetermined dates due to monetary policy interventions and liquidity constraints. This corresponds to stochastic discontinuities (i.e., discontinuities occurring at ex-ante known points in time) in their dynamics.
In this work, we propose a  term structure modelling framework based on overnight rates and characterize absence of arbitrage in a generalised Heath-Jarrow-Morton (HJM) setting. 
We extend the classical short-rate approach to accommodate stochastic discontinuities, developing a tractable setup driven by affine semimartingales. 
In this context, we show that simple specifications allow to capture stylized facts of the jump behavior of overnight rates. In a Gaussian setting, we provide explicit valuation formulas for bonds and caplets. 
Furthermore, we investigate hedging in the sense of local risk-minimization when the underlying term structures feature stochastic discontinuities. 
\end{abstract}

\section{Introduction}

The discontinuation of the publication of Libor rates for the majority of currencies and tenors on January 1, 2022, and the cessation of the US dollar Libor panel on June 30, 2023, mark a major transition for interest rate markets.\footnote{See the FCA announcement on cessation and loss of representativeness of the Libor benchmarks (\url{www.fca.org.uk/publication/documents/future-cessation-loss-representativeness-libor-benchmarks.pdf}) and the FCA decision on the synthetic USD-indexed Libor (\url{www.fca.org.uk/publication/feedback/fs23-2.pdf}).}
In the reform of interest rate benchmarks, {\em overnight rates} play a central role, such as SOFR (secured overnight financing rate) in the US, SONIA (Sterling overnight index average) in the UK and €STR (Euro short-term rate) in the Euro zone, sometimes generically referred to as {\em risk-free rates} (RFRs).

A distinctive feature of overnight rates is the presence of \emph{stochastic discontinuities} in their dynamics: jumps and spikes occurring at predetermined dates or at regular intervals, as a result of monetary policy interventions as well as regulatory and liquidity constraints. In particular, overnight rates tend to  jump in correspondence with meetings of the monetary policy authority, and these meetings usually follow a set schedule.
This is confirmed by the analysis in  \cite{BackwellHayes21} documenting that most of the variation in the SONIA rate over the years 2016-2020 occurs in correspondence with the meeting dates of the Monetary Policy Committee of the Bank of England. 
The recent analysis in \cite{schlogl2023term} provides evidence of a similar phenomenon for SOFR, highlighting the importance of modelling scheduled jumps that coincide with the Federal Open Market Committee (FOMC) meeting dates.

In this work, the predetermined dates at which the overnight rate (and, potentially, forward rates) is expected to exhibit discontinuities will be called {\em expected jump dates} and denoted by $\cS=\{s_1,\ldots,s_M\}$. 
In the case of SOFR the presence of expected jump dates is well illustrated by Figure~\ref{fig1}. As a particular example,  consider the spike observed on September 17, 2019. According to \cite{fed2019}, 
	``Strains in money markets in September seem to have originated from routine market events, including a corporate tax payment date and Treasury coupon settlement. The outsized and unexpected moves in money market rates were likely amplified by a number of factors''.
The analysis of \cite{fed2019} suggests that the date of this spike was known well in advance (a corporate tax payment date coinciding with a Treasury coupon settlement), while the magnitude of the jump was obviously not predictable.

\begin{figure*}
\begin{center}
	\includegraphics[width=10cm]{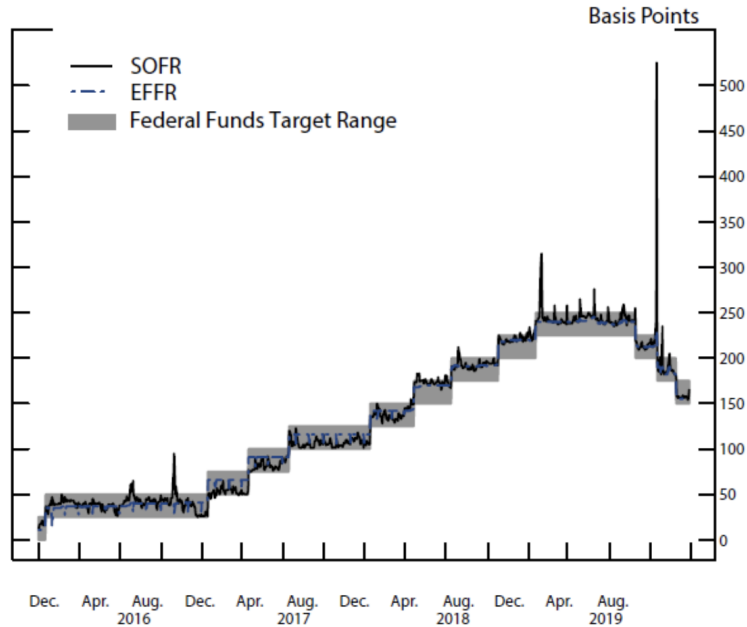}
\end{center}
	\caption{SOFR time series from 01/12/2015 until 30/09/2019. Spikes and jumps occurring at regular points in time are clearly visible. Source: \cite{fed2019}.}\label{fig1}
\end{figure*}

Starting from these observations, we develop  a general framework for interest rate markets in the presence of overnight rates. A first key point is that the natural choice for the num\'eraire asset in this context corresponds to a rolled-over investment in the overnight rate, according to a schedule $\cT=\{t_1,t_2,\ldots\}$ of {\em roll-over dates}. The num\'eraire can therefore exhibit jumps in correspondence with such roll-over dates, which represent an additional source of stochastic discontinuities, besides the expected jump dates mentioned above. Typically, the roll-over dates will be quite frequent, while expected jump dates  will be less frequent, and it is not excluded that some expected jump dates coincide with roll-over dates.
In this context, building on the results  of   \cite{fontana2020term}, we extend the HJM approach by allowing for two different types of stochastic discontinuities. We characterize absence of arbitrage by means of generalized drift conditions, with specific no-arbitrage restrictions related to the stochastic discontinuities.

As a second main contribution, we develop a tractable class of models based on affine semimartingales, i.e.\ affine processes which allow for  stochastic discontinuities (see \cite{keller-ressel2019}). We show that affine semimartingale models for an overnight rate provide a natural extension of classical affine short-rate models to the case of stochastic discontinuities. As illustrated by a simple example, this class of models allows reproducing several stylized features of overnight rates, in particular spikes and jumps at fixed times. Towards practical applications, we derive explicit pricing formulas for bonds and caplets in an extended Hull-White model with discontinuities.

Finally, we study the hedging of derivatives related to risk-free rates (or, more generally, derivatives written on Libor fallbacks determined by RFRs). The presence of stochastic discontinuities induces market incompleteness and, therefore, we resort to local risk-minimization. We show that the locally risk-minimizing strategy admits a  decomposition into two components: a dynamic continuous-time strategy representing the delta-hedging strategy, and an additional component that optimally rebalances the portfolio in correspondence to the expected jump dates.
We exemplify this result by considering the problem of hedging a SOFR-caplet by trading in a SOFR futures contract, the most liquidly traded contract written on SOFR at the time of writing.

\subsection*{Related literature}

The reform of interest rate benchmarks is receiving considerable attention and, therefore, we limit our review of the literature to some contributions that are specifically related to our work, referring to \cite{Henrard19}, \cite{Piterbarg20}, \cite{KlingleraSyrstadb21} and \cite{HugginsSchaller2022} for a general analysis of the challenges of the Libor reform. One of the earliest and most influential contributions is \cite{MercurioLyashenko19}, based on an extension of the Libor market model. Classical short-rate models have been revisited in the post-Libor universe by several authors, starting from \cite{Mercurio18}.
Several recent works employ the classical Hull-White model (see, e.g., \cite{Hofman20}, \cite{Turfus20}, \cite{Hasegawa21}).
Always in a short-rate setup, \cite{SkovSkovmand20}  propose a multi-factor Gaussian model in order to analyze SOFR futures, while \cite{Fontana22} develops a model driven by general affine processes in view of pricing applications. \cite{rutkowski2021pricing} adopt a Vasi\v cek model for SOFR and other reference rates and study the hedging of SOFR-based derivatives, also in the presence of funding costs and collateralization. 
A different approach is taken by  \cite{MacrinaSkovmand20}, who adopt a linear rational model for the savings account and derive several pricing formulas. We also mention \cite{Willems20}, where an extended SABR model is applied to the pricing of caplets in the post-Libor setup.

The papers mentioned in the previous paragraph do not take into account  stochastic discontinuities in the dynamics of overnight rates. To the best of our knowledge, the earliest works acknowledging this fact are \cite{Piazzesi2001, Piazzesi2005}. 
In \cite{KimWright}, a term structure model with jumps occurring in correspondence of macroeconomic announcements dates is presented. In recent works, 
stochastic discontinuities are playing an increasingly important role.
In particular, \cite{AndersenBang20} develop a model that can generate spikes in the SOFR dynamics, both at totally inaccessible times and at anticipated times. \cite{GellertSchloegl21} and \cite{BraceGellertSchloegl22} show that a diffusive HJM model for instantaneous forward rates is compatible with the presence of jumps/spikes at fixed times in the short rate, consistently with the empirical evidence on SOFR.
In \cite{BackwellHayes21}, the SONIA rate is modelled via a short-rate approach by relying on a pure jump process with both unexpected and predetermined jump dates. 
\cite{schlogl2023term} make use of a short-rate affine jump-diffusion framework to provide a  model which is able to jointly fit the overnight US policy rate, SOFR and SOFR futures rates.
\cite{harju2023target} models US overnight rates, such as the effective Fed funds rate and the USD overnight Libor, in the setup of short rate models, also including jumps at fixed times.
Finally, we mention that stochastic discontinuities also play an important role in credit markets (see \cite{GehmlichSchmidt2018} and \cite{FontanaSchmidt2018}), while a general framework for multi-curve markets with stochastic discontinuities is developed in \cite{fontana2020term}.

\subsection*{Structure of the paper} 

In Section \ref{sec2}, we present a general view on post-Libor interest rate markets based on overnight rates. In Section \ref{sec3}, we develop a modelling framework based on the HJM approach extended to the case of stochastic discontinuities. In Section \ref{affine}, we introduce a model based on affine semimartingales for overnight rates and provide explicit valuation formulas in an extended Hull-White model. Finally, in Section \ref{hedging} we study locally risk-minimizing hedging strategies in the presence of stochastic discontinuities. Some of the more technical proofs are postponed to the Appendix.

\section{Interest rate markets with overnight rates} \label{sec2}

In this section, we give a fundamental description of an interest rate market in the presence of overnight rates. We replace the classical assumption of the existence of a savings account generated by an instantaneous short rate by a more general structure, which in particular allows for stochastic discontinuities.
The first key point of our analysis is therefore a systematic study of the implications of a num\'eraire with stochastic discontinuities generated by overnight investment.

\subsection{The num\'eraire}
\label{sec:numeraire}

The num\'eraire asset obtained by investing according to an overnight rate is generated by a roll-over strategy and as such is piecewise constant, being updated at every roll-over date. Those dates are predetermined (hence, deterministic) and we denote them by $t_1<t_2<\cdots$, usually corresponding to business days. We call those dates \emph{roll-over dates}, as mentioned in the introduction, and collect them in the set $\cT:=\{t_n \colon n\in\N\}$. Moreover, we set $t_0:=0$. 

In this situation, the value at time $t\geq0$ of the overnight num\'eraire takes the following form:
\begin{equation}	\label{eq:discrete_num}
\prod_{t_{n+1}\leq t}\bigl(1+r_{t_n}(t_{n+1}-t_n)\bigr),
\end{equation}
with $r_{t_n}$ representing the overnight rate applicable to the time period $[t_n,t_{n+1}]$.  

To allow for greater generality, and in particular to include the classical framework into our setting (see Remark \ref{rem:numeraire} below), we assume that the num\'eraire process $S^0$ is given by
\begin{equation} \label{eq:ass1}
S^0_t := \exp\left( \int_0^t \rho_u \, \eta(du) \right), 
\qquad \text{ for all }t \ge 0,
\end{equation}
where $\eta$ is a Borel measure on $\R_+$ with the following structure:
\begin{align} \label{eq:mu}
\eta(du) = du + \sum_{n\in\N} \delta_{t_n}(du).
\end{align}
Here $\delta_{t_n}$ denotes the Dirac measure in $\{t_n\}$ and $\rho$ is an adapted process satisfying $\int_0^T|\rho_t|\eta(dt)<\infty$ a.s., for all $T>0$. We will refer to $\rho$ as the {\em risk-free rate} (RFR) process.

\begin{remark}	\label{rem:numeraire}
(i) A discretely updated num\'eraire of the form \eqref{eq:discrete_num} can be easily recovered from \eqref{eq:ass1} by taking $\rho_{t_n}=\log(1+r_{t_{n-1}}(t_n-t_{n-1}))$, for all $n\in\N$, and setting $\rho\equiv0$ outside of the set $\cT$.

(ii) Classical interest rate models without stochastic discontinuities are based on the assumption of an absolutely continuous num\'eraire $S^0$. This case can be recovered by simply setting $\cT=\emptyset$, thus yielding $S^0_t=\exp(\int_0^t\rho_u\,du)$ in \eqref{eq:ass1}, with $\rho$ representing the risk-free short rate.
\end{remark}

\subsection{Notions of interest rates}
\label{sec:def_rates}

The second key point in our analysis are the new features of post-Libor markets. We therefore introduce several notions of interest rates which are important in such  markets, relying mostly on \cite{MercurioLyashenko19}.

We denote by $P(t,T)$ the price at date $t$ of a zero-coupon bond with unit payoff at maturity $T \ge t$. In the following, we will refer to zero-coupon bonds simply as bonds. 

\subsubsection{The backward-looking rate}
\label{sec:back}

For each $0\leq S<T$, the {\em setting-in-arrears rate} $R(S,T)$ is the rate that is achieved over the period $[S,T]$ by a rolled-over investment according to the roll-over dates $\cT$. This yields the rate
\begin{equation}	\label{eq:RT_i}
R(S,T) := \frac{1}{T-S} \bigg( \prod_{n\in N(S,T)} \frac{1}{P(t_n,t_{n+1})}-1\bigg),
\end{equation}
where $N(S,T):=\{n\in\N \cup \{0\} : S \leq t_n \text{ and } t_{n+1}\leq T\}$ denotes the set of indices of the roll-over dates $t_n$ for which the interval $[t_n,t_{n+1}]$ is completely contained in $[S,T]$. 

The setting-in-arrears rate $R(S,T)$ is said to be {\em backward-looking}, since its value can be determined only at the end of the accrual period $[S,T]$ and not at its beginning (as it would be the case for a forward rate).
Backward-looking rates play a central role in post-Libor markets, having been adopted as the reference fallback rates for most Libor-based contracts and transactions. 

The generality of our setup enables us to work with the exact definition \eqref{eq:RT_i} of the setting-in-arrears rate regardless of the structure of $S^0$ (in particular, also when $S^0$ is absolutely continuous).

\begin{remark}[Relation to the num\'eraire]
The overnight rate $r_{t_n}$ mentioned in \eqref{eq:discrete_num} can be obtained from bond prices via $1+r_{t_n}(t_{n+1}-t_n)=1/P(t_n,t_{n+1})$. Hence, if $S^0$ is the overnight num\'eraire given in equation \eqref{eq:discrete_num}, then the setting-in-arrears rate $R(S,T)$ can be directly written in terms of $S^0$ as
\begin{equation}	\label{eq:SIA_num}
R(S,T) = \frac{1}{T-S} \bigg( \frac{S^0_T}{S^0_S} -1 \bigg).
\end{equation}
In the literature (see, e.g., \cite{MercurioLyashenko19} and \cite{SkovSkovmand20}), the num\'eraire $S^0$ is usually assumed to be absolutely continuous (see part (ii) of Remark \ref{rem:numeraire}) and equation \eqref{eq:SIA_num} is adopted as an approximation of the setting-in-arrears rate $R(S,T)$. As mentioned above, this approximation is not necessary in our setting with stochastic discontinuities.
\end{remark}

\subsubsection{The forward-looking rate}
\label{sec:fwd}

The {\em forward-looking rate} $F(S,T)$ is defined as the rate $K$ that makes equal to zero the market value at time $S$ of the payoff $(T-S)(R(S,T)-K)$ delivered at maturity $T$, for $0\leq S<T$. Such an agreement is called a {\em single-period swap}. 
In contrast to the backward-looking rate, the forward-looking rate $F(S,T)$ is determined at the beginning of the accrual period.

\begin{remark}[On the notion of forward-looking rate]
\label{CMErate}
The most widely adopted forward-looking rate is the {\em CME term SOFR rate}, which has been approved by the Alternative Reference Rates Committee (ARRC) in 2021 for use in cash products and, with some restrictions, derivatives. Moreover, since June 30, 2023, the CME term SOFR rates (plus the respective ISDA fixed spread adjustment) are used for calculation of the temporary synthetic 1-, 3- and 6-month USD-indexed Libor rates meant to facilitate the transition  of the contracts that reference USD-indexed Libor (see the footnote on page 1). 
In theory, forward-looking rates should be determined as discussed above from market quotes of overnight index swaps (OIS), of which single-period swaps are the basic building blocks. However, the CME term SOFR rate is currently determined with a specific methodology, based on the work of \cite{HeitfieldPark2019}, that relies on market quotes of SOFR futures. This is due to the fact that liquidity in SOFR OIS is not deemed sufficient, while SOFR futures are traded in much larger volumes.\footnote{CME has however announced that SOFR OIS will be used in the calculation of term SOFR rates once their volume will exceed 25\% of the volume of SOFR futures (see \url{https://www.cmegroup.com/market-data/files/cme-term-sofr-reference-rates-benchmark-methodology.pdf}).} Using futures prices to compute forward rates is a model-dependent procedure and also incurs into the issue of convexity adjustments. Therefore, in the current market environment of imperfect liquidity, it may happen that the CME term SOFR rate is not perfectly aligned with the forward-looking rate derived from SOFR swap quotes. 
In this paper, we shall not consider this issue, which will be addressed in a separate work.
\end{remark}

\subsubsection{The forward term rate}	\label{sec:fwd_rates}

Single-period swaps can be considered the basic contracts written on backward-looking and forward-looking rates in post-Libor markets, analogously to forward rate agreements in classical interest rate markets.  
For $0\leq S<T$ and $t\in[0,T]$, the {\em backward-looking forward rate} $R(t,S,T)$ is defined as the rate $K$ that makes equal to zero the value at time $t$ of a single-period swap delivering the payoff $(T-S)(R(S,T)-K)$ at maturity $T$.
Note that, differently from the classical concept of a forward rate, the backward-looking forward rate 
$R(t,S,T)$ is also defined inside the accrual period (i.e., for $t \in [S,T]$), due to the backward-looking nature of $R(S,T)$. 

In an analogous way, we define the {\em forward-looking forward rate} $F(t,S,T)$ for any $t\in[0,S]$ as the rate $K$ that makes equal to zero the value at time $t$ of a single-period swap delivering the payoff $(T-S)(F(S,T)-K)$ at maturity $T$.
Clearly, the forward-looking forward rate satisfies $F(S,S,T)=F(S,T)$, where $F(S,T)$ is the forward-looking rate introduced above.

Comparing the notions of backward-looking and forward-looking forward rate, we notice that
\begin{equation}\label{eq:equality_rates}
R(t,S,T) = F(t,S,T),
\qquad\text{ for all }t\in[0,S].
\end{equation}
However, while the forward-looking forward rate $F(\cdot,S,T)$ stops evolving at time $S$, the backward-looking forward rate $R(\cdot,S,T)$ continues to evolve until time $T$, when it reaches the terminal condition $R(T,S,T) = R(S,T)$.
Identity \eqref{eq:equality_rates}, which has been first pointed out in  \cite{MercurioLyashenko19}, therefore implies that backward-looking forward rates and forward-looking forward rates can be consolidated into a single process $R(\cdot,S,T)$. In this work, we adopt this viewpoint and generically call the process $R(\cdot,S,T)$, considered in the whole time interval $[0,T]$,  the {\em forward term rate}.

\begin{remark}
[On the validity of  \eqref{eq:equality_rates}]
The discussion in Remark \ref{CMErate} implies that, if $F(S,T)$ is assumed to coincide with the CME term SOFR rate, then the identity $F(S,T)=R(S,S,T)$ may fail to hold in the current market environment, due to the CME term SOFR rate calculation methodology. In turn, this implies that violations to \eqref{eq:equality_rates} cannot be excluded a priori. 
Allowing for violations to \eqref{eq:equality_rates} would require to model $F(\cdot,S,T)$ separately from $R(\cdot,S,T)$, with the introduction of an intrinsic multi-curve dimension in the model. We will develop this aspect in a forthcoming work, where we will also provide a mathematical description of violations to \eqref{eq:equality_rates} in terms of the strict local martingale property of solutions to a suitable BSDE.
Here, howeover, we do not account for possible violations to \eqref{eq:equality_rates}, assuming implicitly that liquidity in SOFR swaps is sufficient for a robust determination of forward term rates, coherently with the viewpoint of \cite{MercurioLyashenko19}.
\end{remark}

It is important to note that forward term rates can be expressed in terms of bond prices. Indeed, in view of formula \eqref{eq:RT_i}, the backward-looking rate $R(S,T)$ can be replicated by holding a static position in a bond with maturity $S$ and investing the payoff $P(S,S)=1$ received at time $S$ into the overnight num\'eraire until time $T$. Since $R(S,T)$ represents the floating leg of a single-period swap referencing the backward-looking rate, this implies that the forward term rate can be written as 
\[
R(t,S,T) = \frac{1}{T-S}\left(\frac{P(t,S)}{P(t,T)}-1\right),
\qquad\text{ for all }0\leq t\leq T,
\]
where we extend the definition of bond prices  beyond maturity $S$ by setting
\[
P(t,S) := \frac{P(t,t_{n(t)})}{P(t_{n(t)-1},t_{n(t)})}\prod_{n\in N(S,t)}\frac{1}{P(t_n,t_{n+1})}
\qquad\text{ for }t>S,
\]
with $N(S,t):=\{n\in\N \cup \{0\} \colon S \leq t_n \text{ and } t_{n+1}\leq t\}$ and $n(t):=\inf\{n\in\N_0 \colon t_n> t\}$, for all $t\geq0$.

The analysis developed in this section shows that, together with the num\'eraire $S^0$, the family of bond prices $\{P(\cdot,T); \, T>0\}$ constitutes the fundamental basis of a term structure model for a post-Libor market as considered here. This is the approach that we are going to adopt and develop in the next section, highlighting the role of stochastic discontinuities.

\section{An extended Heath-Jarrow-Morton framework}
\label{sec3}

In this section, we develop a general term structure model based on overnight rates in the presence of stochastic discontinuities. The main result of this section is Theorem \ref{thm:NA}, which provides a set of necessary and sufficient conditions for the risk-neutral property of a given probability measure. 

We recall that, as mentioned in the introduction, in the considered interest rate market two different types of stochastic discontinuities arise naturally:
\begin{enumerate}
\item {\em roll-over dates} $\cT=\{t_1,t_2,\dots\}$, corresponding to the discontinuities in the num\'eraire process $S^0$, encoded in the atoms of the measure $\eta$ introduced in \eqref{eq:mu};
\item {\em expected jump dates} $\cS=\{s_1,\ldots,s_M\}$, representing a set of deterministic times at which the RFR process $\rho$ and forward rates are expected to exhibit jumps.
\end{enumerate}
We do allow for an overlap of these sets (i.e., $\cT\cap\cS\neq\emptyset$), meaning that stochastic discontinuities in the dynamics of the term structure of RFRs can occur simultaneously to some of the roll-over dates. 
In comparison to the roll-over dates in $\cT$, the expected jump dates in $\cS$ are much less frequent, so that we consider only a finite number of them.

\begin{remark}[Extension to predictable times]	\label{rem:pred_times}
The results of this section are also valid in the more general setting where $\cS$ is a countable family of {\em predictable times}, see \cite{FontanaSchmidt2018}. For simplicity of presentation and in order to treat $\cS$ with the same techniques used for $\cT$, we suppose that $\cS$ is a finite family of fixed dates.
\end{remark}

Let $(\Omega,\cF,\FF, Q)$ be a given stochastic basis, endowed with a filtration $\FF=(\cF_t)_{t\geq0}$ satisfying the usual conditions and supporting a $d$-dimensional Brownian motion $W=(W_t)_{t\geq0}$ and an integer-valued random measure $\mu(dt,dx)$ on $\R_+\times E$, with compensator $\nu(dt,dx)=\lambda_t(dx)dt$, where $\lambda_t(dx)$ is a kernel from $(\Omega\times\R_+,\ccP)$ into $(E,\ccB(E))$, with $\ccP$ denoting the predictable sigma-field on $\Omega\times\R_+$ and $(E,\ccB(E))$ a Polish space with its Borel sigma-field. The compensated random measure is denoted by $\tilde{\mu}(dt,dx):=\mu(dt,dx)-\nu(dt,dx)$.
We refer to \cite{JacodShiryaev} for all unexplained notions of stochastic calculus.

The analysis of Section \ref{sec:def_rates} shows that  the key ingredient of a term structure model in the post-Libor framework as considered here is the family of bond prices $\{P(\cdot,T);T>0\}$ together with  the num\'eraire $S^0$ defined in \eqref{eq:ass1}.
To introduce a general  modelling framework we consider an extension of the HJM approach allowing for discontinuous term structures. To this end, we assume that
\begin{equation}
\label{eq:HJM-bond price}
P(t,T) = \exp\bigg( - \int_{(t,T]} f(t,u) \eta(du) \bigg),
\qquad\text{ for all }0\leq t\leq T, 
\end{equation}
with the convention  $\int_{(T,T]} f(T,u) \eta(du) =0$, for all $T\geq0$. 
Moreover, we assume that the  \emph{instantaneous forward rates} $f(\cdot,T)$, for $T\geq0$, are given by
\begin{equation}
\label{eq:fwd_rate}
 f(t,T) = f(0,T) + \int_0^t \alpha(s,T) ds + \int_0^t \varphi(s, T) dW_s + \int_0^t\int_E\psi(s,x,T)\tilde{\mu}(ds,dx) + V(t, T),
\end{equation}	
for $0 \le t \le T$, 
where $V(\cdot,T)$ is a pure jump adapted process such that $\{\Delta V(\cdot,T)\neq0\}\subseteq\Omega\times\cS$. 

\begin{remark}
Integration with respect to the measure $\eta$ in \eqref{eq:HJM-bond price} is justified by the fact that, since the num\'eraire $S^0$ jumps in correspondence of the atoms of $\eta$, bond prices are expected to be discontinuous (in maturity) at those points. More precisely, absence of arbitrage implies that bond prices have necessarily to be of the form \eqref{eq:HJM-bond price} with respect to the same measure $\eta$ appearing in \eqref{eq:ass1}.
This fact has been first pointed out in \cite{GehmlichSchmidt2018}, albeit in the context of default modelling.
\end{remark}

To proceed further, we introduce the following technical requirements on \eqref{eq:fwd_rate}.

\begin{assumption}\label{ass}
The following conditions hold:
\begin{enumerate}[(i)]
\item the {\em initial forward curve} $T\mapsto f(0, T)$ is $(\cF_0\otimes\ccB(\R_+))$-measurable, real-valued and satisfies $\int_0^T|f(0, u)|\eta(du) <  \infty$, for all $T>0$,
\item the {\em drift process} $\alpha:\Omega\times\R^2_+\to\R$ is progressively measurable\footnote{This means that the map $\alpha(\cdot,\cdot)|_{[0,t]}:\Omega\times[0,t]\times\R_+\to\R$ is $(\cF_t\otimes\ccB([0,t])\otimes\ccB(\R_+))$-measurable, for all $t\in\R_+$.}, satisfies $\alpha(t,T)=0$ for  $0\leq T<t$, and
\[
\int_0^T\int_0^u|\alpha(s,u)|ds\,\eta(du)<\infty,
\qquad\text{ for all }T>0,
\]
\item the {\em diffusive volatility process} $\varphi:\Omega \times\R_+^2\to\R^d$ is progressively measurable, satisfies $\varphi(t,T)=0$ for  $0\leq T<t$, and
\[
\sum_{i=1}^d\int_0^T\left(\int_0^u|\varphi^i(s,u)|^2ds\right)^{1/2}\eta(du)<\infty,
\qquad\text{ for all }T>0,
\]
\item the {\em jump function} $\psi:\Omega\times\R_+\times E\times\R_+\to\R$ is a $(\ccP\otimes\ccB(E)\otimes\ccB(\R_+))$-measurable function satisfying $\psi(t,x,T)=0$ for $0\leq T<t$ and $x\in E$, and
\[
\int_0^T\int_E\int_0^T|\psi(s,x,u)|^2\eta(du)\lambda_s(dx) ds<\infty,
\]
\item the {\em stochastic discontinuity process} $V(\cdot,T)=(V(t,T))_{t\in[0,T]}$ satisfies $\int_0^T|\Delta V(s,u)|\eta(du)<\infty$ for all $s\in\cS$ and $\Delta V(t,T)=0$ for all $0\leq T<t$.
\end{enumerate}
\end{assumption}
 		
As a consequence of Assumption \ref{ass}, all integrals appearing in the forward rate equation \eqref{eq:fwd_rate} are well-defined for $\eta$-a.e. $T>0$. The integrability requirements in parts (ii)-(iii)-(iv) of Assumption \ref{ass} ensure that ordinary and stochastic Fubini theorems can be applied, in the versions of \cite[Theorem 2.2]{Veraar12} for Brownian integrals and \cite[Proposition A.2]{BMKR} for stochastic integrals with respect to the compensated random measure $\tilde{\mu}$. 

For all $0 \leq t \leq T$ and $x\in E$, we define
\begin{align*}
\bar \alpha(t,T)  &:= \int_{[t,T]}\alpha(t,u) \eta(du),\\ 
\bar \varphi(t,T)  &:= \int_{[t,T]}\varphi(t,u) \eta(du),\\ 
\bar \psi(t,x,T)  &:= \int_{[t,T]}\psi(t,x,u) \eta(du),\\ 
\bar V(t,T)  &:= \int_{[t,T]}  \Delta V(t,u) \eta(du).
\end{align*}

The probability measure $Q$ is a \emph{risk-neutral} measure if $P(\cdot,T)/S^0$ is a local martingale under $Q$, for every $T>0$.
The existence of a risk-neutral measure suffices to ensure absence of arbitrage, in the sense of {\em no asymptotic free lunch with vanishing risk} (NAFLVR, see \cite{CuchieroKleinTeichmann}), for the large financial market where bonds of all maturities are traded. In turn, this ensures the validity of NAFLVR in the post-Libor market described in Section \ref{sec:def_rates}.
We refer to \cite{KleinSchmidtTeichmann2015} and \cite[Section 6]{fontana2020term} for a detailed analysis of absence of arbitrage in interest rate markets.

In the proof of Theorem \ref{thm:NA} below, we will  use the following stochastic exponential representation of bond prices discounted by the num\'eraire $S^0$. 

\begin{lemma}	\label{lem:bond_stoch_exp}
 Suppose that Assumption \ref{ass} holds. Then, for every $T > 0$,
 \begin{align}
\label{eq:bond_prices_stoch_exp}
\begin{aligned}
\frac{P(\cdot, T)}{S^0} = P(0, T)\mathcal{E} \bigg(
& - \int_0^\cdot \bar{\alpha}(s, T) ds 
+ \frac{1}{2} \int_0^\cdot \|\bar{\varphi}(s, T)\|^2 ds \\
& - \int_0^\cdot \bar{\varphi}(s, T) dW_s 
- \int_0^\cdot \int_E \bar{\psi}(s,x,T)\tilde{\mu}(ds,dx) \\
& + \int_0^\cdot \int_E \bigl(e^{-\bar{\psi}(s,x,T)}-1 + \bar{\psi}(s,x,T)\bigr) \mu(ds,dx) + \int_0^\cdot \bigl(f(u, u)-\rho_u\bigr)du \\
&+\sum_{\tau\in\cS\cup\cT}\big(e^{-\bar{V}(\tau,T)\delta_{\cS}(\tau) + \left(f(\tau,\tau)-\rho_\tau\right)\delta_{\cT}(\tau)}  - 1 \big)    \Ind_{\dbraco{\tau,\infty}} \bigg).
\end{aligned}
\end{align}
\end{lemma}
\begin{proof}
Assumption \ref{ass} ensures that ordinary and stochastic Fubini  theorems are applicable and, hence, we can write
\begin{equation}	\label{eq:stoch_exp_proof}	\begin{aligned}
F(t,T) &:= \int_{(t,T]}f(t,u)\eta(du)
= \int_0^Tf(0,u)\eta(du) + \int_0^t\bar{\alpha}(s,T)ds + \int_0^t\bar{\varphi}(s,T)dW_s \\
&\quad + \sum_{i=1}^M\bar{V}(s_i,T)\Ind_{\{s_i\leq t\}} + \int_0^t\int_E\bar{\psi}(s,x,T)\tilde{\mu}(ds,dx) - \int_0^tf(u,u)\eta(du)
=: G(t,T),
\end{aligned}\end{equation}
for all $0\leq t<T$. Since $F(T,T)=0$, to prove that \eqref{eq:stoch_exp_proof} holds also for $t=T$, it suffices to show that $\Delta G(T,T):=G(T,T)-G(T-,T)=-F(T-,T)$. Moreover, since $\mu(\{T\},E)=0$ a.s. for all $T\in\R_+$, it is enough to show that $\Delta G(T,T)=-F(T-,T)$ for all $T\in\cS\cup\cT$. This holds because 
\begin{align*}
\Delta G(T,T) &= \bar{V}(T,T)\delta_{\cS}(T)-f(T,T)\delta_{\cT}(T)
= \bigl(\bar{V}(T,T)\delta_{\cS}(T)-f(T,T)\bigr)\delta_{\cT}(T)	
= -f(T-,T)\delta_{\cT}(T) \\
&= -F(T-,T),
\end{align*}
as a consequence of the definition of $F(t,T)$.
Therefore,
\[
\frac{P(t,T)}{S^0_t} = \exp\left(-G(t,T)-\int_0^t\rho_u\,\eta(du)\right).
\]
Representation \eqref{eq:bond_prices_stoch_exp} then follows from \cite[Theorem II.8.10]{JacodShiryaev}.
\end{proof}

The following theorem characterizes when $S^0$-discounted bond prices are local martingales under $Q$ or, equivalently, when $Q$ is a risk-neutral measure.

\begin{theorem}	\label{thm:NA}
Suppose that Assumption \ref{ass} holds. Then, $Q$ is a risk-neutral measure if and only if 
\begin{equation}
\label{int-cond-1}
\int_0^T \int_E \big| e^{-\bar{\psi}(s,x,T)}-1 + \bar{\psi}(s,x,T))  \big|\lambda_s(dx)ds <   \infty\text{ a.s.}\quad\text{ for all }T\geq0,
\end{equation}
and the random variable 
\begin{equation}
\label{int-cond-2}
e^{-\int_{(\tau,T]}\Delta V(\tau,u)\eta(du)\delta_{\cS}(\tau)-\Delta\rho_{\tau}\delta_{\cT}(\tau)}
\end{equation}
is sigma-integrable with respect to $\ccF_{\tau-}$, for every $\tau\in\cS\cup\cT$, and the following four conditions hold:
\begin{enumerate}
\item $f(t,t)=\rho_t$ $(Q\otimes dt)$-a.e.,
\item on $[0,T]$, it holds $(Q\otimes dt)$-a.e. that 
\[
\bar{\alpha}(t,T) = \frac{1}{2}\|\bar{\varphi}(t,T)\|^2 + \int_E\bigl(e^{-\bar{\psi}(t,x,T)}-1+\bar{\psi}(t,x,T)\bigr)\lambda_t(dx),
\]
\item for every $j=1,\ldots,N$, it holds a.s. 
that
\[
f(t_j-,t_j) = \rho_{t_j-} -\log\bigl(E[e^{-\Delta\rho_{t_j}}|\cF_{t_j-}]\bigr),
\]
\item for every $i=1,\ldots,M$, it holds a.s. 
that
\[
E\Bigl[e^{-\Delta\rho_{s_i}\delta_{\cT}(s_i)}\bigl(e^{-\int_{(s_i,T]}\Delta V(s_i,u)\eta(du)}-1\bigr)\Big|\cF_{s_i-}\Bigr] = 0.
\]
\end{enumerate}
\end{theorem}
\begin{proof}
In view of \eqref{eq:bond_prices_stoch_exp}, $P(\cdot,T)/S^0$ is a local martingale if and only if the finite variation process 
\begin{equation}	\label{eq:K}
K(\cdot,T) := \int_0^{\cdot}k(s,T)ds + K^{(1)}(\cdot,T) + K^{(2)}(\cdot,T)
\end{equation}
is a local martingale, where
\begin{align*}
k(t,T) &:= f(t,t)-\rho_t-\bar{\alpha}(t,T)+\frac{1}{2}\|\bar{\varphi}(t,T)\|^2 ,\\
K^{(1)}(t,T) &:= \int_0^t \int_E \bigl(e^{-\bar{\psi}(s,x,T)}-1 + \bar{\psi}(s,x,T)\bigr) \mu(ds,dx), \\
K^{(2)}(t,T) &:= \sum_{\tau\in\cS\cup\cT}\big(e^{-\bar{V}(\tau,T)\delta_{\cS}(\tau) + \left(f(\tau,\tau)-\rho_\tau\right)\delta_{\cT}(\tau)}  - 1 \big)    \Ind_{\{\tau\leq t\}},
\end{align*}
for all $t\in[0,T]$. If $K(\cdot,T)$ is a local martingale, it is also of locally integrable variation by \cite[Lemma I.3.11]{JacodShiryaev}. Since $\{\Delta K^{(1)}(\cdot,T)\neq0\}\cap\{\Delta K^{(2)}(\cdot,T)\neq0\}=\emptyset$, it holds that $|\Delta K^{(i)}(\cdot,T)|\leq|\Delta K(\cdot,T)|$, for $i=1,2$, and, therefore, both processes $K^{(i)}(\cdot,T)$, $i=1,2$, are of locally integrable variation. This implies that condition \eqref{int-cond-1} holds and the random variable
\[
e^{-\bar{V}(\tau,T)\delta_{\cS}(\tau) + \left(f(\tau,\tau)-\rho_\tau\right)\delta_{\cT}(\tau)}
\]
is sigma-integrable with respect to $\cF_{\tau-}$, for all $\tau\in\cS\cup\cT$ (see  \cite[Theorem 5.28]{he2018semimartingale}). 
Taking into account the definition of $\bar{V}(t,T)$, equation \eqref{eq:fwd_rate} and the fact that $(f(\tau-,\tau)-\rho_{\tau-})\delta_{\cT}(\tau)$ is $\cF_{\tau-}$-measurable, the latter property is equivalent to the sigma-integrability of the random variable \eqref{int-cond-2}.
Denoting by $\widehat{K}^{(i)}(\cdot,T)$ the compensator of $K^{(i)}(\cdot,T)$, for $i=1,2$, and making use of \cite[Theorem 5.29]{he2018semimartingale}, the local martingale property of $K(\cdot,T)$ is then equivalent to the validity of
\begin{align}
0 &= \int_0^{\cdot}k(s,T)ds+\widehat{K}^{(1)}(\cdot,T)
= \int_0^{\cdot}\left(k(s,T)+\int_E\bigl(e^{-\bar{\psi}(s,x,T)}-1 + \bar{\psi}(s,x,T)\bigr)\lambda_s(dx)\right)ds,
\label{eq:comp1}\\
0 &= \widehat{K}^{(2)}(\cdot,T) =  \sum_{\tau\in\cS\cup\cT}\left(E\bigl[\big(e^{-\bar{V}(\tau,T)\delta_{\cS}(\tau) + \left(f(\tau,\tau)-\rho_\tau\right)\delta_{\cT}(\tau)} \big)\big|\cF_{\tau-}\bigr] -1\right)   \Ind_{\dbraco{\tau,\infty}},
\label{eq:comp2}
\end{align}
up to an evanescent set. Equation \eqref{eq:comp1} holds if and only if
\begin{equation}	\label{eq:comp3}
f(t,t)-\rho_t-\bar{\alpha}(t,T)+\frac{1}{2}\|\bar{\varphi}(t,T)\|^2 + \int_E \bigl(e^{-\bar{\psi}(t,x,T)}-1 + \bar{\psi}(t,x,T)\bigr) \lambda_t(dx) = 0
\end{equation}
outside of a set of $(Q\otimes dt)$-measure zero. Taking $T=t$ in \eqref{eq:comp3} yields condition (i), while condition (ii) follows by inserting condition (i) into \eqref{eq:comp3}.
In view of \eqref{eq:fwd_rate} and the definition of $\bar{V}(t,T)$, equation \eqref{eq:comp2} holds if and only if 
\begin{equation}	\label{eq:comp4}
e^{(f(\tau-,\tau)-\rho_{\tau-})\delta_{\cT}(\tau)}E\Bigl[e^{-\int_{(\tau,T]}\Delta V(\tau,u)\eta(du)\delta_{\cS}(\tau)-\Delta\rho_\tau\delta_{\cT}(\tau)}\Big|\cF_{\tau-}\Bigr]=1
\text{ a.s.}
\end{equation}
for all $\tau\in\cS\cup\cT$.
Taking $T=\tau$ in \eqref{eq:comp4} yields condition (iii), while condition (iv) is then obtained by inserting condition (iii) into \eqref{eq:comp4}.

Conversely, if condition \eqref{int-cond-1} is satisfied and the random variable in \eqref{int-cond-2} is sigma-integrable, for every $j=1,\ldots,M$, then the compensator $\widehat{K}^{(i)}(\cdot,T)$ is well-defined, for $i=1,2$. It is then straightforward to verify that if conditions (i)--(iv) are satisfied, then equations \eqref{eq:comp1}-\eqref{eq:comp2} hold true. This implies that the process $K(\cdot,T)$ given in \eqref{eq:K} is a local martingale, thus proving the local martingale property of $P(\cdot,T)/S^0$, for every $T\geq0$.
\end{proof}

In the absence of stochastic discontinuities, it is well-known that $P(\cdot,T)/S^0$ is a local martingale if and only if the integrability requirement \eqref{int-cond-1} and conditions (i)-(ii) of Theorem \ref{thm:NA} hold, see \cite[Proposition 5.3]{BMKR}.
The presence of stochastic discontinuities is reflected in conditions (iii)-(iv) (together with the requirement of sigma-integrability of \eqref{int-cond-2}). More specifically, condition (iii) relates the stochastic discontinuities $\cT$ in the num\'eraire to the short end of the forward rate, while condition (iv) concerns the stochastic discontinuities $\cS$ in the forward term rate. 
Taken together, conditions (iii)-(iv) exclude the possibility of predicting the  size (or the direction) of the jump occurring at any discontinuity date, as this would be incompatible with absence of arbitrage (see \cite{FPP19} for an analysis of the arbitrage possibilities arising with predictable jumps).

The conditions of Theorem \ref{thm:NA} admit a simplification under the following additional assumption.
\begin{assumption}	\label{ass:jumps}
The set $\{(\omega,t)\in\Omega\times\cT:\Delta\rho_t(\omega)\neq0\}$ is evanescent and $\cS\cap\cT=\emptyset$.
\end{assumption}
This assumption  corresponds to requiring that the RFR $\rho$ and the forward rates do not jump in correspondence of the roll-over dates of the  num\'eraire $S^0$. In this case, conditions (i) and (iii) can be rewritten in the following compact way: 
\[
f(t,t)=\rho_t \quad (Q\otimes\eta)\text{-a.e.}
\]
Moreover, under Assumption \ref{ass:jumps} the term $\Delta\rho_{s_i}\delta_{\cT}(s_i)$ in condition (iv) vanishes.

\begin{remark}[Relation to \cite{fontana2020term}]
The presence of the two distinct sets $\cS$ and $\cT$ of discontinuity dates distinguishes the present setup from the framework used for modelling multi-curve term structures in \cite{fontana2020term}. For this reason, we have given full self-contained proofs of Lemma \ref{lem:bond_stoch_exp} and Theorem \ref{thm:NA}, which cannot be deduced in a straightforward way from \cite{fontana2020term}. At the same time, we have simplified some of the original techniques of \cite{fontana2020term}. 
\end{remark}

The next corollary provides the dynamics of the short end $f(t,t)$ of the instantaneous forward rate. We omit the proof, which follows the same arguments of \cite[Proposition 6.1]{Filipovic2009}. 

\begin{corollary}
Suppose that Assumption \ref{ass} holds.
Assume furthermore that $f(0, T)$, $\alpha(t, T)$, $\varphi(t, T)$, $\psi(t, x, T)$ are differentiable in $T$ with $\int_0^T|\partial_uf(0,u)|du<\infty$, for all $T\geq0$, and such that conditions (ii), (iii), (iv) from Assumption \ref{ass} are satisfied for $ \partial_T \alpha(t, T),  \partial_T  \varphi(t, T),  \partial_T  \psi(t, x, T)$, respectively. Then it holds that 
\[
 f(t,t) = f(0,0) +  \int_0^t \zeta(s) ds +   \int_0^t \varphi(s, s) dW_s + \int_0^t\int_E\psi(s,x,s)\tilde{\mu}(ds,dx)
 +  V(t, t),
\]
for all $t\geq0$, where
\begin{equation*}
\zeta(s) := \partial_s f(0, s) + \alpha(s,s) + \int_0^s \partial_s \alpha(u, s) du + \int_0^s \partial_s \varphi(u, s) dW_u +  \int_0^s \int_E\partial_s \psi(u,x,s) \tilde{\mu}(du,dx)
\end{equation*}
and 
$V(t, t) = \sum_{j=1}^{M} \Delta V(s_j, t) \Ind_{\{s_j \leq t\}}. $
\end{corollary}

\begin{remark}
In the classical HJM setup, the short end of the instantaneous forward rate is equal to the short rate, i.e., we have $r_t = f(t,t)$. In the presence of stochastic discontinuities, this must not necessarily be the case and absence of arbitrage implies only  the  equality  $f(t,t)=\rho_t$ $(Q\otimes dt)$-a.e., as shown in Theorem  \ref{thm:NA} (i).
It is interesting to see what happens if one assumes that the RFR $\rho$ is defined as the short end of the instantaneous forward rate by imposing  
\begin{equation}
\label{eq:short-end}
\rho_t:=f(t,t), 
\qquad\text{ for all } t \geq 0. 
\end{equation}
Condition (i) of Theorem \ref{thm:NA} is of course automatically satisfied.
Conditions (iii)-(iv) are equivalent to the validity of condition \eqref{eq:comp2} in the proof of Theorem \ref{thm:NA}. Under \eqref{eq:short-end}, this condition becomes
\[
E\bigl[e^{-\int_{[s_i,T]}\Delta V(s_i,u)\eta(du)}\big|\cF_{s_i-}\bigr]  = 1,
\qquad\text{ for all }i=1,\ldots,M.
\]
This in particular entails that, if all stochastic discontinuities in the model correspond only to roll-over dates $\cT$ (meaning that $\cS=\emptyset$), then setting $\rho$ equal to the short end of the forward rate as in \eqref{eq:short-end} implies that conditions (i), (iii) and (iv) of Theorem \ref{thm:NA} hold and $Q$ is a risk-neutral measure if and only if condition (ii) is satisfied, as in the classical HJM setup. 
\end{remark}

To illustrate the flexibility of the  extended HJM approach, we show how a generalized version of the popular Cheyette model with stochastic discontinuities fits into our framework.
In Section \ref{affine}, we will develop an approach based on the philosophy of short-rate modelling.

\begin{example}[A Cheyette-style model]
The Cheyette model is a Gaussian, multi-factor, Markovian model introduced in \cite{Cheyette}. This model class was independently considered also in \cite{jamshidian1991bond}, \cite{babbs1993generalized}, \cite{ritchken1993finite}. We refer to \cite{Beyna2013} for a detailed presentation of this class of models. 

In this example, we illustrate an extension of the Cheyette model to a setup with stochastic discontinuities. We consider only the presence of the expected jump dates $\cS$, assuming  $\cT=\emptyset$ for simplicity of presentation, so that $S^0=\exp(\int_0^\cdot r_u\, du)$.
We specify as follows the forward rates:
\[
f(t,T) = f(0,T) + \int_0^t \alpha(s,T) ds + \int_0^t \varphi(s,T) dW_s + \sum_{s_i \le t} \big(\alpha_i(T) + \xi_i g_i(T)\big), 
\]
where $\alpha$, $\varphi$, $\alpha_i$, $g_i$, for $i=1,\dots,M$, are deterministic functions and $\xi_i\sim \cN(\mu_i,\sigma_i^2)$, for $i=1,\dots,M$, are independent normally distributed random variables which are also independent of $W$. 

The central assumption underlying the Cheyette model corresponds to  
\[
\varphi(t,T) = \frac{a(T)}{a(t)} b(t),
\]
for suitable functions $a$ and $b$. This is a so-called separable volatility  structure which allows separating $t$ and $T$ in the continuous part of the model. More precisely, denoting by $f^c(\cdot,T)$ the continuous part of the process $f(\cdot,T)$, we have that
\begin{equation}    \label{eq:Cheyette_c}
f^c(t,T) = f(0,T) + \frac{a(T)}{a(t)} X^c_t + U^c(t,T),
\end{equation}
where $X^c$ is a mean-reverting Gaussian continuous Markov process with volatility function $b$ and $U^c(t,T)$ is an explicitly known function, see \cite[Section 2.2]{Beyna2013}.

Analogously, for the discontinuous part $f^d(t,T)$ of the forward rate it is natural to assume that
\[
g_i(T) = B_i A(T),
\qquad\text{ for }i=1,\ldots,M,
\]
for suitable constants $B_1,\ldots,B_M$ and an integrable function $A$. Similarly as above, this separable structure leads to a representation of $f^d(t,T)$ analogous to \eqref{eq:Cheyette_c}. Indeed, defining the pure jump Markov process $J:= \sum_{s_i \le \cdot} \xi_i B_i$, we obtain that
\begin{equation}    \label{eq:Cheyette_d}
f^d(t,T) := \sum_{s_i \le t} \big(\alpha_i(T) + \xi^i g_i(T)\big) 
= U^d(t,T) + A(T) J_t,
\end{equation}
where the explicit form of the function $U^d(t,T)$ can be deduced from condition (iv) of Theorem \ref{thm:NA}: 
\[
U^d(t,T) = \sum_{s_i\leq t}B_iA(T)\bigl(\sigma^2_iB_i\bar{A}(s_i,T)-\mu_i\bigr),
\]
with $\bar{A}(s_i,T):=\int_{s_i}^TA(u)du$, for all $i=1,\ldots,M$ and $T\geq0$.

If in addition $a(\cdot)=A(\cdot)$, then combining equations \eqref{eq:Cheyette_c} and \eqref{eq:Cheyette_d} we directly obtain that
\[
f(t,T) = f(0,T) + \frac{a(T)}{a(t)}X_t + U(t,T),
\]
where $X_t:=X^c_t+a(t)J_t$ and $U(t,T):=U^c(t,T)+U^d(t,T)$. Moreover, making use of equation (2.20) in \cite{Beyna2013} and by means of straightforward computations, it can be proved that $X$ is a mean-reverting Gaussian Markov jump-diffusion process with speed of mean reversion $\partial_t\log(a(t))$, jumping only in correspondence with the predetermined stochastic discontinuities dates $s_1,\ldots,s_M$.
This example illustrates that the Cheyette model class can be easily extended to the case of stochastic discontinuities, retaining a remarkable level of analytical tractability.
\end{example}

\section{The affine framework for overnight rate modelling}
\label{affine}

In this section, we present a modelling framework based on {\em affine semimartingales}. We focus here on the direct modelling of the overnight rate to provide a setup for developing   extensions of the Hull-White and other affine models which are able to include stylized facts of overnight markets, namely spikes and jumps occurring at predetermined dates. Note that an affine specification of the HJM framework presented in Section \ref{sec3} is also possible, leading to explicit conditions in  Theorem \ref{thm:NA}. 

Affine semimartingales, as introduced in \cite{keller-ressel2019}, generalize affine processes by allowing for discontinuities at fixed points in time with possibly state-dependent jump sizes and are therefore tailor-made to interest rate markets in the presence of overnight rates. 
We shortly introduce some general notions on affine semimartingales and then develop an affine model for overnight rates.

\subsection{Affine semimartingales}
Consider a semimartingale $X$ taking values in $D:=\R^m_+\times\R^n$. It is called \emph{affine} when  for all $0\leq t\leq T<\infty$ and $u\in\mathcal{U}:=\C^m_-\times{\rm i}\R^n$,
\begin{equation}	\label{eq:charfun}
E\bigl[e^{\langle u,X_T\rangle} | \cF_t\bigr] 
= \exp\bigl(\phi_t(T,u) + \langle\psi_t(T,u),X_t\rangle\bigr),
\end{equation}
where the functions $\phi_t(T,u)$ and $\psi_t(T,u)$ take values in $\C$ and $\C^d$, respectively. 
We assume that ${\rm conv}({\rm supp}(X_t))=D$, for all $t>0$, and that $X$ is quasi-regular and infinitely divisible, in the sense that the regular conditional distribution $Q(X_t\in dx|X_s)$ is an infinitely divisible probability measure on $D$ a.s. for all $0\leq s\leq t<\infty$. 

The property that affine processes have affine semimartingale characteristics has a natural extension to the semimartingale case. More precisely, by \cite[Lemma 4.3 and Theorem 3.2]{keller-ressel2019}, there is no loss of generality in assuming that $X$ is a Markov process and the semimartingale characteristics $(B^X,C^X,\nu^X)$ of $X$ with respect to a fixed truncation function $h:\R^d\rightarrow\R^d$, with $d=m+n$, have the following structure, where $B^{X,c}$ and $\nu^{X,c}$ denote the continuous parts of $B^X$ and $\nu^X$, respectively:
\begin{align*}
\nonumber B^{X,c}_t &= \int_0^t\Bigl(\beta^X_0(s) + \sum_{i=1}^dX^i_{s-}\beta^X_i(s)\Bigr)ds,	\\
\nonumber C^X_t &= \int_0^t\Bigl(\alpha^X_0(s) + \sum_{i=1}^dX^i_{s-}\alpha^X_i(s)\Bigr)ds,	\\
\nonumber \nu^{X,c}(dt,dx) &= \Bigl(\mu^X_0(t,dx)+\sum_{i=1}^dX^i_{t-}\mu^X_i(t,dx)\Bigr)dt,	\\
\int_D(e^{\langle u,x\rangle}-1)\nu^X(\{t\},dx) &= \exp\Bigl(\gamma^X_0(t,u)+\sum_{i=1}^dX^i_{t-}\gamma^X_i(t,u)\Bigr)-1,
\end{align*}
where $\beta^X_i:\R_+\rightarrow\R^d$, $\alpha^X_i:\R_+\rightarrow\mathcal{S}^d$,  $\gamma^X_i:\R_+\times\mathcal{U}\rightarrow\C$, for $i=0, 1, \ldots, d$, and $(\mu^X_i(t,\cdot))_{t\geq0}$ is a family of  L\'evy  measures on $D\setminus\{0\}$, with $\mathcal{S}^d$ denoting the cone of symmetric positive semidefinite $d\times d$ matrices. 
In addition, it holds that $\gamma^X_i(t,u)=0$, for all $i$ and  for all $(t,u)\in(\R_+\setminus J^X)\times\mathcal{U}$, where $J^X:=\{t\in\R_+ : \nu^X(\{t\},D)\neq0\}$ represents the set of stochastic discontinuity dates of the affine semimartingale $X$. 

The functions $\phi$ and $\psi$ from equation \eqref{eq:charfun} can now be described  by the following generalized Riccati equations, see \cite[Theorem 3.2]{keller-ressel2019}: first, for all $(T,u)\in\R_+\times\mathcal{U}$, the continuous parts $\phi^c_t(T,u)$ and $\psi^c_t(T,u)$ satisfy
\begin{align*}
\frac{d\phi_t^c(T,u)}{dt} &= -F^X(t,\psi_t(T,u)),	\\
\frac{d\psi_t^c(T,u)}{dt} &= -R^X(t,\psi_t(T,u)),	
\end{align*}
where the functions $F^X$ and $R^X$ are of L\'evy-Khintchine form, for all $i=1,\ldots,d$:
\begin{align*}
F^X(t,u) &= \langle\beta^X_0(t),u\rangle + \frac{1}{2}\langle u,\alpha^X_0(t)u\rangle + \int_{D \setminus\{0\}}\bigl(e^{\langle x,u\rangle}-1-\langle h(x),u\rangle\bigr)\mu^X_0(t,dx),	\\
R^X_i(t,u) &= \langle\beta^X_i(t),u\rangle + \frac{1}{2}\langle u,\alpha^X_i(t)u\rangle + \int_{D\setminus\{0\}}\bigl(e^{\langle x,u\rangle}-1-\langle h(x),u\rangle\bigr)\mu^X_i(t,dx).
\end{align*}
Second, for all $(T,u)\in\R_+\times\mathcal{U}$, the discontinuous parts of $\phi_t(T,u)$ and $\psi_t(T,u)$ are determined by
\begin{align*}
\Delta\phi_t(T,u) &= -\gamma^X_0(t,\psi_t(T,u)),	\\
\Delta\psi_t(T,u) &= -\bar{\gamma}^X(t,\psi_t(T,u)),
\end{align*}
where $\bar{\gamma}^X(t,u)=(\gamma^X_1(t,u),\ldots,\gamma^X_d(t,u))\in\R^d$. Finally, the terminal conditions are given by
\[
\phi_T(T,u) = 0
\qquad\text{ and }\qquad
\psi_T(T,u) = u.
\] 

\begin{remark}
In order to be consistent with the framework of Section \ref{sec3}, we restrict our attention to affine semimartingales whose characteristics do not contain singular continuous parts. However, we point out that all results of this section can be generalized in a straightforward way to that case. 
\end{remark}

\subsection{Affine models for overnight rates} \label{sec:affine_short}
A widely used approach in interest rate theory  consists in modelling  the short rate and computing bond and derivative prices by risk-neutral valuation. If the short-rate model is sufficiently tractable (e.g., in the case of affine models as in \cite[Section 13]{DuffieFilipovicSchachermayer}), then explicit valuation formulas for bond prices and interest rate derivatives can be obtained. In this section, we show how this short-rate approach can be extended to the case of affine semimartingales with stochastic discontinuities.

We assume that the RFR rate $\rho$ is given by
\begin{equation}	\label{eq:SOFR_affine}
\rho_t = \ell(t) + \langle\Lambda,X_t\rangle,
\qquad t\geq0,
\end{equation}
where $\Lambda\in\R^d$ and $\ell:\R_+\rightarrow\R$ is a deterministic and c\`adl\`ag function such that $\int_0^T|\ell(u)|\eta(du)<\infty$, for all $T>0$, with measure $\eta$ from \eqref{eq:mu}. The function $\ell$ serves to fit the initially observed term structure of bond prices, similarly as in \cite{BrigoMercurio01b}.

In view of \eqref{eq:charfun}-\eqref{eq:SOFR_affine}, the conditional characteristic function of the RFR is directly available. However, this does not suffice for the valuation of bonds and interest rate derivatives, since the discount factor and most types of payoffs depend on the integrated process $R:=\int_0^{\cdot}\rho_t\,\eta(dt)$. 
The following proposition shows that the joint process $(X,R)$ is an affine semimartingale on the extended state space $D\times\R$. 
This result, which can be considered of independent interest, represents a generalization to affine semimartingales of the enlargement of the state space approach of \cite[Section 11.2]{DuffieFilipovicSchachermayer}. The proof is technical and therefore postponed to the Appendix.

\begin{proposition}	\label{prop:Y}
Consider the $(D\times\R)$-valued process $Y:=(X,R)$. Then, 
\begin{equation}	\label{eq:CF_Y}
E\bigl[e^{\langle u,X_T\rangle+vR_T} | \cF_t\bigr] 
= \exp\bigl(\Phi_t(T,u,v) + \langle\Psi_t(T,u,v),X_t\rangle+vR_t\bigr),
\end{equation}
for all $(u,v)\in\mathcal{U}\times{\rm i}\R$ and $0\leq t\leq T$, where the functions $\Phi_t(T,u,v)$ and $\Psi_t(T,u,v)$ take values in $\C$ and $\C^d$, respectively, and solve the following generalized Riccati equations:
\begin{equation}	\label{eq:RicattiY_1}	\begin{aligned}
\frac{d\Phi_t^c(T,u,v)}{dt} &= -F^Y(t,\Psi_t(T,u,v),v),	\\
\frac{d\Psi_t^c(T,u,v)}{dt} &= -R^Y(t,\Psi_t(T,u,v),v),	
\end{aligned}	\end{equation}
with
\begin{equation}	\label{eq:RicattiY_2}	\begin{aligned}
F^Y(t,u,v) &= F^X(t,u) + \ell(t)v,	\\
R^Y_i(t,u,v) &= R^X_i(t,u) + \Lambda_iv,
\end{aligned}	\end{equation}
for all $i=1,\ldots,d$, and
\begin{equation}	\label{eq:RicattiY_3}	\begin{aligned}
\Delta\Phi_t(T,u,v) 
&= -\delta_{\cT^c}(t)\gamma^X_0(t,\Psi_t(T,u,v)) - \delta_{\cT}(t)\bigl(v\ell(t)+\gamma^X_0(t,\Psi_t(T,u,v)+\Lambda v)\bigr),	\\
\Delta\Psi_t(T,u,v) 
&= -\delta_{\cT^c}(t)\bar{\gamma}^X(t,\Psi_t(T,u,v)) - \delta_{\cT}(t)\bigl(v\Lambda+\bar{\gamma}^X(t,\Psi_t(T,u,v)+\Lambda v)\bigr),
\end{aligned}	\end{equation}
 satisfying the terminal conditions
\begin{equation}	\label{eq:RicattiY_4}	\begin{aligned}
\Phi_T(T,u,v) = 0
\qquad\text{ and }\qquad
\Psi_T(T,u,v) = u.
\end{aligned}	\end{equation}
In particular, the joint process $(X,R)$ is an affine infinitely divisible semimartingale.
\end{proposition}

The explicit characterization of the conditional Fourier transform of the joint process $(X,R)$ obtained in Proposition \ref{prop:Y} allows for an efficient valuation of a variety of interest rate derivatives in the post-Libor environment. In particular, bonds can be priced by evaluating \eqref{eq:CF_Y} at $(u,w)=(0,-1)$, whenever the expectation is finite. In turn, this leads to explicit pricing formulas for all linear derivatives such as swaps. Non-linear derivatives can be priced by relying on Fourier methods, analogously to the case of affine processes, see \cite[Section 10.3]{Filipovic2009} and \cite{Fontana22}.

\subsection{A Hull-White model for the overnight rate}
\label{sec:Hull_White}

In this section, we present a tractable specification of an affine model for the overnight rate.
More specifically, we generalize the Hull-White model to the case of stochastic discontinuities. 
We consider a finite set $\cT=\{t_1,\ldots,t_N\}$ representing the roll-over dates of the num\'eraire $S^0$, while $S=\{s_1,\ldots,s_M\}$ denotes the set of expected jump dates of the RFR process.

Let $\rho$ be the unique strong solution of the following stochastic differential equation:
\begin{align}\label{eq:drho}
	d\rho_t = \big( \alpha(t) + \beta \rho_t\big) dt + \sigma dW_t + dJ_t,
\end{align}
where $\alpha:\R\to\R$ is a continuous function,  
$\beta \in \R$,  $\sigma \ge 0$ and $W=(W_t)_{t\geq0}$ is a Brownian motion.  
The process $J=(J_t)_{t\geq0}$ in \eqref{eq:drho} is a pure jump process specified as follows:
\begin{equation}	\label{eq:process_J}
J = \sum_{i=1}^M \xi^i\Ind_{\dbraco{s_i,\infty}},
\end{equation}
where the random variables $\{\xi^i;i=1,\ldots,M\}$ are assumed to be independent of $W$.
The following lemma gives the explicit solution to \eqref{eq:drho} and is an immediate consequence of It\^o's formula.

\begin{lemma} \label{lem:rho}
The solution $\rho=(\rho_t)_{t\geq0}$ of the SDE \eqref{eq:drho} satisfies, for all $0\leq t\leq T$,
\begin{equation}\label{eq:335}
		\rho_T = \rho_te^{\beta(T-t)} + \int_t^T e^{\beta(T-s)} \alpha(s) ds + \sigma e^{\beta T}\int_t^T e^{-\beta s} dW_s
		+ \sum_{i=1}^M \ind{s_i \in (t,T]} \xi^i e^{\beta(T-s_i)}.
\end{equation}
\end{lemma}

As illustrated in the following example, processes of the form \eqref{eq:335}, despite their simple structure, allow for different types of stochastic discontinuities that are in line with the empirical features of overnight rates, as discussed in the introduction (see Figure \ref{fig1}).  

\begin{example}[A two-factor affine specification]\label{ex:affine}
Overnight rates are characterized by the presence of spikes and jumps, both potentially occurring at predetermined dates.
While affine processes can exhibit a spiky behavior (e.g., shot-noise processes, see \cite{Schmidt2008modelling}), they do not accommodate spikes and jumps at predetermined dates. These features can be reproduced by affine semimartingales of the form \eqref{eq:335}. For the purpose of illustration, let $\rho^1$ and $\rho^2$ be two independent process satisfying \eqref{eq:335}, with parameters
\begin{gather*}
\rho^1_0 = 0.01875, \quad
\beta^1 = 0.20, \quad
\alpha^1(t)=\alpha^1=0.01, \quad
\sigma^1 = 0.012,\\
\rho^2_0 = 0, \quad
\beta^2 = 80, \quad
\alpha^2(t)=\alpha^2=0, \quad
\sigma^2 = 0,
\end{gather*}
for all $t\geq0$. In this example, the process $\rho^1$ is assumed to have a single discontinuity at time $s^1_1=150$, while $\rho^2$ has discontinuity dates $s^2_1=50$ and $s^2_2=100$. We assume that all jump sizes are i.i.d. and distributed as $\ccN(0.1,0.4^2)$. Figure \ref{fig:vasicek} shows a simulated path of the process $\rho:=\rho^1+\rho^2$. We can observe that $\rho$ exhibits spikes at $t=50$ and $t=100$ and a jump to a new level at $t=150$.
The spiky behavior is generated by $\rho^2$, which has no diffusive component and very high mean-reversion speed (increasing $\beta^2$ further would produce even more pronounced spikes). The jump to a new level at $t=150$ is generated by the component $\rho^1$, which has a much slower mean-reversion.
\end{example}

\begin{figure*}[t]
\includegraphics[width=8cm]{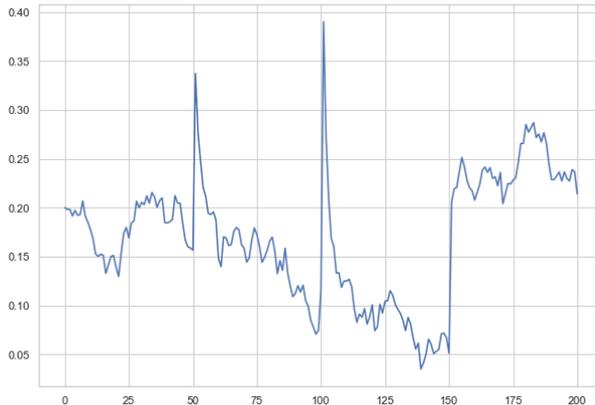}
\caption{A simulation of the two-factor affine model considered in Example \ref{ex:affine}. The simulated path  exhibits spikes at the discontinuity dates $t=50$ and $t=100$ and a jump to a new level at $t=150$.}
\label{fig:vasicek}
\end{figure*}

In the next results, we shall make use of the following notation:
\[
a(t,T) := \int_t^T e^{\beta(T-s)} \alpha(s) ds. 
\]
For $t=0$, we simply write $a(T):=a(0,T)$.
In the following proposition, we compute the conditional characteristic function of the solution to SDE \eqref{eq:335}.
While this result can be deduced from the general theory of affine semimartingales, we provide a direct proof exploiting the structure of \eqref{eq:335}.

\begin{proposition}
Let $\rho=(\rho_t)_{t\geq0}$ be the solution of the SDE \eqref{eq:335} and assume that the random variables $\{\xi^i;i=1,\ldots,M\}$ are independent. Then, for every $u\in{\rm i}\R$ and $0\leq t\leq T$, it holds that
\[
E\big[e^{u \rho_T} \big| \rho_t \big] 
= \exp\big( \phi(u,t,T) + \psi(u,T-t) \rho_t \big),
\]
where 
\begin{align*}
\psi(u,T-t) &= u e^{\beta(T-t)}, \\
\phi(u,t,T) &=  \ u \, a(t,T)+ \frac{ (u \sigma)^2 }{4 \beta} \big( e^{2 \beta(T-t)} -1\big) + \sum_{i:s_i \in (t,T]} \log\Big( E\Big[ \exp( ue^{\beta(T-s_i)} \xi^i)\Big] \Big). 
\end{align*}
\end{proposition}
\begin{proof}
From representation \eqref{eq:335}, we directly obtain that
\begin{align*}
E\big[e^{u \rho_T} \big|\rho_t \big] &=
\exp\big(  u e^{\beta(T-t)} \rho_t  \big) \cdot \exp\bigg(  u \int_t^T e^{\beta(t-s)} \alpha(s) ds + \frac{u^2\sigma^2}{4 \beta} \Big(  e^{2\beta(T-t)}-1 \Big)  \bigg)\notag \\
&\quad \cdot E\bigg[\exp\bigg( u \sum_{i=1}^M \ind{s_i \in (t,T]} \xi^i e^{\beta(T-s_i)} \bigg) \bigg]. 
\end{align*}
This immediately yields that $ \psi(u,T-t) = u \exp(\beta(T-t))$. The result of the proposition  follows by relying on the independence of the random variables  $\{\xi^i;i=1,\ldots,M\}$.
\end{proof}

Explicit expressions for bond prices can be obtained either by applying the general result from Proposition \ref{prop:Y} or by a direct computation of the conditional characteristic function of the time-integral $\int_0^T \rho_t\,\eta(dt)$, as illustrated in the next subsection under a specific assumption on the distribution of the jump sizes.

\subsection{A Gaussian Hull-White model for the overnight rate}  \label{sec:gaussian}

Until the end of this section we assume that the random variables $\{\xi^i\colon i=1,\ldots,M\}$ are independent and normally distributed. In this case, we immediately see from \eqref{eq:335} that $\rho$ is a Gaussian process (more precisely, a Markov process with Gaussian increments). Moreover, the random variable $R_T:=\int_0^T \rho_t \,\eta(dt)$ is also normally distributed for all $T>0$. This immediately yields the following proposition.
 
\begin{proposition}
\label{prop:bond-price-Gaussian}
For all $0\leq t \leq T$, it holds that
\begin{align*}
P(t,T) = \exp\Big(  - E[R_T|\rho_t,R_t] +R_t  + \frac 1 2 \Var(R_T | \rho_t, R_t )  \Big).
\end{align*}
\end{proposition}

We proceed to compute explicitly the conditional mean and variance of $R_T$ appearing in Proposition \ref{prop:bond-price-Gaussian}.  The proofs of the following two lemmata are based on rather lengthy computations and, therefore, are deferred to the Appendix. As a first step, we compute the conditional mean and covariance of $\rho$.
We denote $m_i := E[\xi_i]$ and $\gamma_i^2 := \Var(\xi_i)$, for each $i=1,\ldots,M$. 

\begin{lemma}
\label{lemma:A1}
For all $0\leq t \leq T$ and  $0\leq t \leq T_1, T_2$, let us denote 	 
\begin{align*}
m(t,T)  & := E[\rho_T| \rho_t] \quad  \textrm{and} \quad c(t,T_1,T_2)  := \Cov(\rho_{T_1},\rho_{T_2} |\rho_t).  
\end{align*}
Then, it holds that
\begin{align} 	\label{mean1}
\begin{aligned}
m(t,T) &  = \rho_t e^{\beta(T-t)} + a(t,T)  +\sum_{i=1}^M \ind{s_i \in (t,T]} m_i e^{\beta(T-s_i)}, \\
c(t,T_1,T_2)&= \frac{\sigma^2e^{\beta(T_1+T_2)}}{2 \beta}\Big( e^{-2\beta t} - e^{- 2 \beta(T_1 \wedge T_2)} \Big)  +\sum_{i=1}^M \ind{s_i \in (t,T_1 \wedge T_2] } \gamma_i^2 e^{ \beta (T_1 + T_2-2s_i)}.	
\end{aligned}
\end{align}
\end{lemma}

By integrating equation \eqref{eq:335} and applying Fubini theorems, we obtain
\begin{equation} \label{eq:RS_explicit}
R_T = R_t + \rho_t B'(t,T) + A'(t,T) + \sigma\int_{t}^{T}B'(s,t)dW_s   
 + \sum_{i=1}^M\ind{s_i\in(t,T]}\bigg(B'(s_i,T)+{\sum_{j=1}^N\ind{s_i=t_j}}\bigg)\xi_i,
\end{equation}
for all $0\leq t\leq T$, where we make use of the notation
\begin{align}
B'(t,T):=\int_{(t,T]}e^{\beta(u-t)}\eta(du)
&= \frac{e^{\beta(T-t)}-1}{\beta}+\sum_{j=1}^N\ind{t_j\in(t,T]}e^{\beta(t_j-t)} \nonumber \\
&=: B(T-t)+\sum_{j=1}^N\ind{t_j\in(t,T]}e^{\beta(t_j-t)} \label{eq:B}
\end{align}
and
\begin{align*}
A'(t,T):=\int_{(t,T]}a(t,u)\eta(du)
&= \int_t^Ta(t,u)du + \sum_{j=1}^N\ind{t_j\in(t,T]}a(t,t_j) \\
&=: A(t,T) + \sum_{j=1}^N\ind{t_j\in(t,T]}a(t,t_j).
\end{align*}

By relying on \eqref{eq:RS_explicit} and on the properties of the random variables $\{\xi_i\colon i=1,\ldots,M\}$, the next lemma gives explicit expressions for the conditional expectation and variance of $R_T$. 
We denote 
\[
\bar{B}(t,T) := 
\frac{e^{2\beta (T-t)}- 1}{2\beta} + \sum_{j=1}^N \ind{t_j \in (t,T]} e^{2\beta (t_j-t)},
\qquad\text{ for all }0\leq t \leq T.
\]

\begin{lemma}
\label{lemma:A2,3}
For all $0\leq t\leq T$, it holds that
	\begin{align}
		E[R_T | \rho_t,R_t] 
		&  = R_t 
		 + \rho_t \left(B(T-t)+\sum_{j=1}^N \ind{t_j \in (t,T]}e^{\beta(t_j-t)}\right)  \notag \\
		 &+ A(t,T) + \sum_{j=1}^N \ind{t_j \in (t,T]}  a(t,t_j) \notag\\
		 &+	\sum_{i=1}^M \ind{s_i \in (t,T]} m_i\left(B(T-s_i)+\sum_{j=1}^N \ind{t_j\in[s_i,T]} e^{\beta(t_j-s_i)}\right),  \label{eq25:app} \\
		 \Var(R_T | \rho_t, R_t ) &=		
\frac{\sigma^2}{2 \beta} B'(t,T)^2
- \frac{\sigma^2}{\beta}\left( \frac{B'(t,T)-(T-t)}{\beta}+\sum_{j=1}^N\ind{t_j\in(t,T]}\Bigl(B'(t_j,T)+\frac{1}{2}-\frac{1}{\beta}\Bigr)\right)	\notag\\
& + 2\sum_{i=1}^M\ind{s_i\in(t,T]}\gamma_i^2\left(\frac{\bar{B}(s_i,T)-B'(s_i,T)}{\beta}+\sum_{j=1}^N\ind{t_j\in[s_i,T]}e^{2\beta(t_j-s_i)}\Bigl(B'(t_j,T)+\frac{1}{2}\Bigr)\right). \notag
\end{align}
\end{lemma} 

\begin{remark}[Impact of expected jump dates]
While $\xi_i$ is the jump in the overnight rate at date $s_i$, its impact on bond prices is modulated by $B(T-s_i)+\sum_{j=1}^N \ind{t_j \in (s_i,T]}e^{\beta(t_j-s_i)}$ in \eqref{eq25:app}, with $B(T-s_i)$ defined in \eqref{eq:B}. In particular, since bond prices determine forward term rates, this allows generating forward term rates that exhibit less pronounced stochastic discontinuities than the overnight rate, in line with the empirical evidence in post-Libor markets based on overnight rates.
\end{remark}

\begin{remark}
In the present affine setup, the availability of exponentially affine  formulas for bond prices implies that the backward-looking rate $R(S,T)$ admits a tractable representation, being defined in \eqref{eq:RT_i} as a product of bond prices. This highlights the fact that our setting enables us to work with the exact definition of the backward-looking rate, while retaining complete analytical tractability.
\end{remark}

We close this section by deriving an explicit formula for the price of a caplet in the context of the present Gaussian Hull-White model. 
In the post-Libor universe one may consider two distinct types of caplets, depending on whether the payoff is determined by the backward-looking rate $R(S,T)$ or by the forward-looking rate $F(S,T)$, see for instance \cite{MercurioLyashenko19,Fontana22}. For illustration, we consider here a forward-looking caplet, whose payoff at date $T$ is given by
\[
H = (T-S)\bigl(F(S,T)-K\bigr)^+,
\]
for some $K>0$. We recall from Section \ref{sec:fwd_rates} that $F(S,T)=R(S,S,T)=(1/P(S,T)-1)/(T-S)$.
In view of Proposition \ref{prop:bond-price-Gaussian} and Lemma \ref{lemma:A2,3}, we have that:
\begin{equation}	\label{eq:bond_repr}
P(t,S) 
= e^{-\rho_t B'(t,S) - \Xi(t,S)},
\qquad\text{ for all }0\leq t\leq S,
\end{equation}
where, for brevity of notation, $\Xi(t,S)$ collects all terms appearing in Lemma \ref{lemma:A2,3} that do not multiply $\rho_t$.
For the determination of the caplet price, we need to compute the $\cF_t$-conditional distribution of $\rho_S$ under the $S$-forward measure $Q^S$ defined by $dQ^S/dQ=1/(S^0_SP(0,S))$. 

\begin{lemma}	\label{lem:fwd_meas}
Under the $S$-forward measure $Q^S$, the following hold:
\begin{enumerate}
\item the process $W^S=(W^S_t)_{t\in[0,S]}$ defined by $W^S_t:=W_t+\sigma\int_0^tB'(u,S)du$, for all $t\in[0,S]$, is a Brownian motion;
\item for each $i=1,\ldots,M$, it holds that
\[
\xi_i\sim\ccN\left(m_i-\ind{s_i\leq S}\gamma^2_i\Bigl(B(S-s_i)+\sum_{j=1}^N\ind{t_j\in[s_i,S]}e^{\beta(t_j-s_i)}\Bigr),\gamma^2_i\right).
\]
Moreover, the random variables $\{\xi_i \colon i=1,\ldots,M\}$ are  mutually independent and independent of the Brownian motion $W^S$.
\end{enumerate}
\end{lemma}
\begin{proof}
Part (i) follows by a standard application of Girsanov's theorem. To prove part (ii), it suffices to consider an arbitrary $\lambda\in\R^M$ and compute the joint characteristic function of the random variables $\{\xi_i\colon i=1,\ldots,M\}$ under the measure $Q^S$ conditionally on the sigma-field $\cF^W:=\sigma(W^S_u;u\in[0,T])$. 
Observe that $\cF^W=\sigma(W_u;u\in[0,T])$.
Denoting by $E^S$ the expectation under $Q^S$ and making use of equation \eqref{eq:RS_explicit}, it holds that
\begin{align*}
E^S\bigl[e^{\im\sum_{i=1}^M\lambda_i\xi_i}|\cF^W\bigr]
&= \frac{E\bigl[e^{-R_S+\im\sum_{i=1}^M\lambda_i\xi_i}|\cF^W\bigr]}{E[e^{-R_S}|\cF^W]}	\\
&= \prod_{i=1}^M\frac{E\bigl[e^{(\im\lambda_i-\ind{s_i\leq S}(B(S-s_i)+\sum_{j=1}^N\ind{t_j\in[s_i,S]}e^{\beta(t_j-s_i)}))\xi_i}\bigr]}{E\bigl[e^{-\ind{s_i\leq S}(B(S-s_i)+\sum_{j=1}^N\ind{t_j\in[s_i,S]}e^{\beta(t_j-s_i)})\xi_i}\bigr]}	\\
&= \prod_{i=1}^Me^{\im\lambda_i(m_i-\ind{s_i\leq S}\gamma^2_i(B(S-s_i)+\sum_{j=1}^N\ind{t_j\in[s_i,S]}e^{\beta(t_j-s_i)}))-\lambda_i^2\frac{\gamma^2_i}{2}},
\end{align*}
where we have used the independence of the random variables $\{\xi_i \colon i=1,\ldots,M\}$ together with the fact that $\xi\sim\ccN(m_i,\gamma_i^2)$, for each $i=1,\ldots,M$.
\end{proof}

As a consequence of Lemma \ref{lem:fwd_meas} and Proposition \ref{lem:rho}, it holds that
\begin{equation}	\label{eq:rho_fwd}
\rho_S\sim\ccN\bigl(\rho_te^{\beta(S-t)}+\Gamma_1(t,S),\Gamma_2(t,S)\bigr)
\qquad\text{ under }Q^S,
\end{equation}
conditionally on $\cF_t$, for $t\in[0,S]$, where the quantities $\Gamma_1(t,S)$ and $\Gamma_2(t,S)$ are defined as
\begin{align*}
\Gamma_1(t,S) &:= \int_t^S\bigl(\alpha(s)-\sigma B'(s,S)\bigr)e^{\beta(S-s)}ds    \\
&\quad +\sum_{i=1}^M\ind{s_i\in(t,S]}\Bigl(m_i-\gamma^2_i\Bigl(B(S-s_i)+\sum_{j=1}^N\ind{t_j\in[s_i,S]}e^{\beta(t_j-s_i)}\Bigr)\Bigr)e^{\beta(S-s_i)},	\\
\Gamma_2(t,S) &:= \sigma^2\int_t^Se^{2\beta(S-s)}ds
+\sum_{i=1}^M\ind{s_i\in(t,S]}\gamma^2_ie^{2\beta(S-s_i)}
= \sigma^2\frac{e^{2\beta(S-t)}-1}{2\beta}+\sum_{i=1}^M\ind{s_i\in(t,S]}\gamma^2_ie^{2\beta(S-s_i)}.
\end{align*}

The following proposition gives the arbitrage-free price of a forward-looking caplet. We denote by $\Phi$ the cumulative distribution function of a $\ccN(0,1)$ random variable.

\begin{proposition}	\label{prop:caplet_price}
Consider a forward-looking caplet delivering payoff $H=(T-S)(F(S,T)-K)^+$ at date $T$, for $S\in[0,T]$ and $K>0$. Its risk-neutral price $H_t$ at time $t\in[0,S]$ is given by
\[
H_t = G(\rho_t,t,S,T,K),
\]
where the function $G(\cdot,t,S,T,K)$ is given for all $x\in\R$ by
\[
G(x,t,S,T,K)
= e^{-xB'(t,S)-\Xi(t,S)}\left(\Phi(d_1(x))-K'e^{-\Xi(S,T)-B'(S,T)(xe^{\beta(S-t)}+\Gamma_1(t,S)-\frac{B'(S,T)\Gamma_2(t,S)}{2})}\Phi(d_2(x))\right),
\]
with $K':=1+(T-S)K$ and
\[
d_1(x) := \frac{-\log(K')+\Xi(S,T)}{B'(S,T)\sqrt{\Gamma_2(t,S)}}+\frac{xe^{\beta(S-t)}+\Gamma_1(t,S)}{\sqrt{\Gamma_2(t,S)}},
\qquad
d_2(x) := d_1(x)-B'(S,T)\sqrt{\Gamma_2(t,S)}.
\]
\end{proposition}
\begin{proof}
Using the definition of the $S$-forward measure $Q^S$ and equation \eqref{eq:bond_repr}, we can compute
\begin{align*}
H_t
&= (T-S)E\left[\frac{S^0_t}{S^0_T}\bigl(F(S,T)-K\bigr)^+\bigg|\cF_t\right]	\\
&= E\left[\frac{S^0_t}{S^0_T}\left(\frac{1}{P(S,T)}-K'\right)^+\bigg|\cF_t\right]	
= K'E\left[\frac{S^0_t}{S^0_S}\left(\frac{1}{K'}-P(S,T)\right)^+\bigg|\cF_t\right]	\\
&= K'P(t,S)E^S\left[\left(\frac{1}{K'}-P(S,T)\right)^+\bigg|\cF_t\right]	
= K'P(t,S)E^S\left[\left(\frac{1}{K'}-e^{-\rho_SB'(S,T)-\Xi(S,T)}\right)^+\bigg|\cF_t\right].
\end{align*}
Under the $S$-forward measure $Q^S$, the $\cF_t$-conditional distribution of $\rho_S$ is given by \eqref{eq:rho_fwd}.
The result then follows by an application of the Black-Scholes formula.
\end{proof}

\section{Hedging in the presence of stochastic discontinuities}\label{hedging}

The presence of stochastic discontinuities may induce market incompleteness, in the sense that perfect replication of payoffs by means of self-financing strategies is not always possible. This is for instance the case of the affine model of Section \ref{sec:Hull_White}, which is affected by the jump risk generated by the process $J$. In this section, we aim at determining optimal hedging strategies in the sense of {\em local risk-minimization}. This corresponds to attaining perfect replication of payoffs while relaxing the self-financing requirement and minimizing the cost of the strategy according to a quadratic criterion (see \cite{Pham00} and \cite{Schweizer01} for an overview of the theory).
In Section \ref{sec:LRM} we provide a general description of local risk-minimization with stochastic discontinuities, while in Section \ref{sec:LRM_example} we study an explicit example in the context of the Gaussian Hull-White model of Section \ref{sec:gaussian}.

\subsection{Local risk-minimization with stochastic discontinuities}\label{sec:LRM}

In order to reduce the technicalities in the presentation and to focus on the impact of stochastic discontinuities, we assume the validity of the following assumption.
We consider a finite time horizon $T$.

\begin{assumption}\label{ass:MRP_simple}
There exists a family $(\xi_1,\ldots,\xi_M)$ of random variables on $(\Omega,\cF,Q)$ taking values in a measurable space $(B,\ccB(B))$ such that  $\xi_i$ is $\cF_{s_i}$-measurable, for each $i=1,\ldots,M$, and every local martingale $N=(N_t)_{t\in[0,T]}$ on $(\Omega,\FF,Q)$ admits a representation of the following form:
\begin{equation}	\label{eq:MRP_simple}
N = N_0 + \int_0^{\cdot}\theta_td W_t + \sum_{i=1}^Mf_i(\xi_i)\Ind_{\dbra{s_i,T}},
\end{equation}
where $\theta\in\Lloc([0,T])$ and $f_i(\cdot):\Omega\times B\to\R$ is a $(\cF_{s_i-}\otimes\ccB(B))$-measurable function such that $E[f_i(\xi_i)|\cF_{s_i-}] = 0$ a.s., for each $i=1,\ldots,M$.
\end{assumption}

Assumption \ref{ass:MRP_simple} is for instance satisfied in the model of Section \ref{sec:Hull_White} if the filtration $\FF$ is generated by the pair $(W,J)$. Note  that the assumption that the discontinuity dates $\cT$ do not appear in the martingale representation \eqref{eq:MRP_simple} is only made for simplicity of presentation.

We suppose that the market contains a traded security with $S^0$-discounted price process $X=(X_t)_{t\in[0,T]}$, assumed to be a special semimartingale with canonical decomposition
\begin{equation}	\label{eq:can_dec}
X = X_0 + A + M,
\end{equation}
where $A=(A_t)_{t\in[0,T]}$ is a predictable  process of finite variation and $M=(M_t)_{t\in[0,T]}$ a square-integrable martingale, with $A_0=M_0=0$. The process $X$ can represent for instance the price process of a SOFR future contract, at present the most liquid product referencing SOFR (see Section \ref{sec:LRM_example}).
Note also that in this section we do not necessarily assume that $Q$ is a risk-neutral measure.

As a consequence of Assumption \ref{ass:MRP_simple}, the martingale $M$ admits a representation of the form
\begin{equation}	\label{eq:MR}
M = \int_0^{\cdot}\eta_u dW_u 
+ \sum_{s_i\leq\cdot}\Delta M_{s_i},
\end{equation}
where $\eta=(\eta_t)_{t\in[0,T]}$ is a predictable process such that $E[\int_0^T\eta^2_udu]<\infty$ and $\Delta M_{s_i}=w_i(\xi_i)$, where the function $w_i$ is as in Assumption \ref{ass:MRP_simple}, for each $i=1,\ldots,M$.
We furthermore assume that $X$ has non-vanishing volatility, in the sense that $\eta_t>0$ a.s. for all $t\in[0,T]$.

By absence of arbitrage, there exists a predictable process $\lambda=(\lambda_t)_{t\in[0,T]}$ such that $A=\int_0^{\cdot}\lambda_ud\langle M\rangle_u$. In particular, this implies that $\Delta A_{s_i}=\lambda_{s_i}E[(\Delta M_{s_i})^2|\cF_{s_i-}]$, for all $i=1,\ldots,N$. We furthermore assume that the expected mean-variance tradeoff is finite, i.e., $E[\int_0^T\lambda^2_ud\langle M\rangle_u]<\infty$. This corresponds to assuming that $X$ satisfies the  {\em structure condition} (see \cite{Schweizer01}).

Let $H$ be a square-integrable $\cF_T$-measurable random variable, representing a discounted payoff. By market incompleteness, $H$ may not be attainable by self-financing trading. We then consider non-self-financing strategies attaining the payoff $H$, as formalized in the following definition, where we denote by $\Theta$ the set of all predictable processes $\zeta=(\zeta_t)_{t\in[0,T]}$ such that $E[\int_0^T\zeta^2_ud\langle M\rangle_u+(\int_0^T|\zeta_udA_u|)^2]<\infty$.

\begin{definition}	\label{def:strategy}
We call {\em $H$-admissible strategy} a pair $\varphi=(\zeta,V)$, where $\zeta=(\zeta_t)_{t\in[0,T]}\in\Theta$ and $V=(V_t)_{t\in[0,T]}$ is an adapted square-integrable process such that $V_T=H$ a.s.
We say that an $H$-admissible strategy $\varphi=(\zeta,V)$ is {\em locally risk-minimizing} if the associated cost process
\[
C_t(\varphi) := V_t-\int_0^t\zeta_u dX_u,
\qquad\text{ for all }t\in[0,T],
\]
is a square-integrable martingale strongly orthogonal to $M$. 
\end{definition}

In Definition \ref{def:strategy}, $\zeta_t$ and $V_t$ represent respectively the positions held in the traded security and the portfolio value at time $t$, for all $t\in[0,T]$.
By \cite[Theorem 3.3]{Schweizer01}, the definition of locally risk-minimizing strategy adopted in Definition \ref{def:strategy} is equivalent to the original definition of \cite{Schweizer91} if the process $A$ in \eqref{eq:can_dec} is continuous, as in the case of the example considered in Section \ref{sec:LRM_example}. For general $A$, Definition \ref{def:strategy} corresponds to the so-called {\em pseudo-locally risk-minimizing} strategy.

In view of \cite[Proposition 3.4]{Schweizer01}, finding a locally risk-minimizing strategy $\varphi=(\zeta,V)$ corresponds to obtaining a decomposition of the payoff $H$ of the form
 \begin{equation}	\label{eq:FS_dec}
H = H_0 + \int_0^T\zeta^H_udX_u + L^H_T,
\end{equation}
where $\zeta^H=(\zeta^H_t)_{t\in[0,T]}\in\Theta$ and $L^H=(L^H_t)_{t\in[0,T]}$ is a square-integrable martingale strongly orthogonal to $M$ with $L^H_0=0$. Decomposition \eqref{eq:FS_dec} is known as the  F\"ollmer-Schweizer decomposition of $H$ and a locally risk-minimizing strategy is then given by $(\zeta^H,V^H)$, where $V^H:=H_0 + \int_0^{\cdot}\zeta^H_udX_u + L^H$.

Under Assumption \ref{ass:MRP_simple}, we can explicitly derive decomposition \eqref{eq:FS_dec} for a generic discounted payoff $H$.
To this effect, let us define $\widehat{Z}:=\ccE(-\int_0^{\cdot}\lambda_udM_u)$ and assume that $\widehat{Z}$ is a strictly positive square-integrable martingale under $Q$. This enables us to define the minimal martingale measure $\widehat{Q}$ by $d\widehat{Q}=\widehat{Z}_TdQ$. 
We can then define the $\widehat{Q}$-martingale $\widehat{H}=(\widehat{H}_t)_{t\in[0,T]}$ by 
\[
\widehat{H}_t:=\widehat{E}[H|\cF_t],
\qquad\text{ for all }t\in[0,T],
\] 
where we denote by $\widehat{E}$ the expectation with respect to $\widehat{Q}$. By Bayes' formula, it holds that $\widehat{H}=N/\widehat{Z}$, with $N_t:=E[\widehat{Z}_TH|\cF_t]$, for all $t\in[0,T]$. As a consequence of Assumption \ref{ass:MRP_simple}, it holds that
\begin{equation}	\label{MRP2}
N = N_0 + \int_0^{\cdot}\theta_udW_u + \sum_{s_i\leq\cdot}\Delta N_{s_i},
\end{equation}
where $\theta\in L^2_{\rm loc}([0,T])$. 
We are now in a position to state the following theorem.

\begin{theorem}	\label{prop:LRM}
Suppose that Assumption \ref{ass:MRP_simple} holds and assume that $\widehat{Z}$ as defined above is a strictly positive square-integrable martingale under $Q$. Let $H$ be an $\cF_T$-measurable random variable and suppose that $\sup_{t\in[0,T]}\widehat{H}_t\in L^2(Q)$. Define the predictable process $\zeta^H=(\zeta^H_t)_{t\in[0,T]}$ by
\begin{equation}	\label{eq:LRM_strategy}
\zeta^H_t := \big(\widehat{Z}^{-1}_{t-}\eta_t^{-1}\theta_t+\widehat{H}_{t-}\lambda_t\big)\delta_{\cS^c}(t)
+ \frac{E[\Delta\widehat{H}_{t}\Delta M_{t}|\cF_{t-}]}{E[(\Delta M_{t})^2|\cF_{t-}]}\delta_{\cS}(t).
\end{equation}
If $\zeta^H\in\Theta$, then an $H$-admissible locally risk-minimizing strategy is given by $\varphi^H=(\zeta^H,V^H)$, where $V^H_t = \widehat{H}_t$, for all $t\in[0,T]$.
\end{theorem}
\begin{proof}
By the product rule, it holds that
\[
\widehat{H}_t = \widehat{Z}^{-1}_tN_t
= \widehat{H}_0 + \int_0^t\widehat{Z}_{u-}^{-1}dN_u + \int_0^tN_{u-}d\widehat{Z}_u^{-1} + [\widehat{Z}^{-1},N]_t,
\]
for all $t\in[0,T]$.
An application of It\^o's formula yields
\[
\widehat{Z}^{-1} = \ccE\left(\int_0^{\cdot}\lambda_u\eta_udW_u + \int_0^{\cdot}\lambda^2_u\eta^2_udu 
+ \sum_{s_i\leq\cdot}\frac{\lambda_{s_i}\Delta M_{s_i}}{1-\lambda_{s_i}\Delta M_{s_i}}\right).
\]
Therefore, in view of equation \eqref{MRP2}, we can compute
\begin{align*}
\widehat{H}_t 
&= \widehat{H}_0 
+ \int_0^t\widehat{Z}^{-1}_{u-}\theta_udW_u 
+ \int_0^tN_{u-}\widehat{Z}^{-1}_{u-}\lambda_u\eta_udW_u 
+ \int_0^tN_{u-}\widehat{Z}^{-1}_{u-}\lambda^2_u\eta^2_udu
+ \int_0^t\widehat{Z}^{-1}_{u-}\theta_u\lambda_u\eta_udu	\notag\\
&\quad+ \sum_{s_i\leq t}\left(\widehat{Z}^{-1}_{s_i-}\Delta N_{s_i}+N_{s_i-}\widehat{Z}^{-1}_{s_i-}\frac{\lambda_{s_i}\Delta M_{s_i}}{1-\lambda_{s_i}\Delta M_{s_i}}+\widehat{Z}^{-1}_{s_i-}\frac{\lambda_{s_i}\Delta M_{s_i}\Delta N_{s_i}}{1-\lambda_{s_i}\Delta M_{s_i}}\right)	\\
&= \int_0^t\widehat{Z}^{-1}_{u-}\big(\theta_u+N_{u-}\lambda_u\eta_u\big)\big(dW_u+\lambda_u\eta_udu\big)
+\sum_{s_i\leq t}\Delta\widehat{H}_{s_i}.
\end{align*}
Since $A^c=\int_0^{\cdot}\lambda_td\langle M\rangle_t^c=\int_0^{\cdot}\lambda_t\eta^2_tdt$ and $\{\Delta X\neq0\}\subseteq\Omega\times\cS$, we have that
\begin{equation}	\label{eq:LRM_proof}
H = \widehat{H}_T = \widehat{H}_0 + \int_0^T\zeta^H_udX_u + \sum_{s_i\leq T}\big(\Delta\widehat{H}_{s_i}-\zeta^H_{s_i}\Delta X_{s_i}\big),
\end{equation}
where $\zeta^H=(\zeta^H_t)_{t\in[0,T]}$ is defined as in \eqref{eq:LRM_strategy}.
We proceed to show that \eqref{eq:LRM_proof} provides the F\"ollmer-Schweizer decomposition \eqref{eq:FS_dec} of $H$, where $L^H:=\sum_{s_i\leq\cdot}(\Delta\widehat{H}_{s_i}-\zeta^H_{s_i}\Delta X_{s_i})$ is a square-integrable martingale strongly orthogonal to $M$ under $Q$.
To prove that $L^H$ is a martingale, it suffices to verify that $E[\Delta L^H_{s_i}|\cF_{s_i-}]=0$ a.s. for all $i=1,\ldots,M$. To this effect, using \eqref{eq:can_dec} we can compute
\begin{align*}
E[\Delta\widehat{H}_{s_i}|\cF_{s_i-}] - \zeta^H_{s_i}E[\Delta X_{s_i}|\cF_{s_i-}]
&= E[\Delta\widehat{H}_{s_i}|\cF_{s_i-}] - \zeta^H_{s_i}\Delta A_{s_i}	\\
&= E[\Delta\widehat{H}_{s_i}|\cF_{s_i-}] - \zeta^H_{s_i}\lambda_{s_i}E[(\Delta M_{s_i})^2|\cF_{s_i-}]	\\
&= E[(1-\lambda_{s_i}\Delta M_{s_i})\Delta\widehat{H}_{s_i}|\cF_{s_i-}]	\\
&= \widehat{E}[\Delta\widehat{H}_{s_i}|\cF_{s_i-}] = 0,
\end{align*}
where in the last step we used the fact that $1-\lambda_{s_i}\Delta M_{s_i}=\widehat{Z}_{s_i}/\widehat{Z}_{s_i-}$ and the $\widehat{Q}$-martingale property of $\widehat{H}$.
Square-integrability of $L^H$ under $Q$ follows from the assumptions.
Finally, $L^H$ and $M$ are strongly orthogonal under $Q$ if and only if $[L^H,M]$ is a $Q$-martingale. In turn, the latter property is equivalent to $E[\Delta L^H_{s_i}\Delta M_{s_i}|\cF_{s_i-}]=0$ a.s., for all $i=1,\ldots,M$. This can be shown to hold since
\begin{align*}
E[\Delta L^H_{s_i}\Delta M_{s_i}|\cF_{s_i-}]
&= E[\Delta\widehat{H}_{s_i}\Delta M_{s_i}|\cF_{s_i-}]
-\zeta^H_{s_i}E[\Delta X_{s_i}\Delta M_{s_i}|\cF_{s_i-}]	\\
&= E[\Delta\widehat{H}_{s_i}\Delta M_{s_i}|\cF_{s_i-}]
-\zeta^H_{s_i}E[(\Delta M_{s_i})^2|\cF_{s_i-}] = 0. \qedhere
\end{align*}
\end{proof}

Theorem \ref{prop:LRM} provides an explicit description of the locally risk-minimizing strategy for a generic payoff $H$. In particular, formula \eqref{eq:LRM_strategy} shows that the locally risk-minimizing strategy consists in a perfect replication at all times $t\in[0,T]\setminus\cS$, when the only active source of randomness is the Brownian motion $W$. The first term on the right-hand side of \eqref{eq:LRM_strategy} corresponds to the Delta-hedging continuous strategy. On the other hand, in correspondence of the expected jump dates $\cS=\{s_1,\ldots,s_M\}$, the strategy $\zeta^H_{s_i}$ is determined by a linear regression of $\Delta\widehat{H}_{s_i}$ onto $\Delta X_{s_i}$, conditionally on $\cF_{s_i-}$. Indeed, we have that
\begin{equation}	\label{eq:regr}
\zeta^H_{s_i} = \frac{\Cov(\Delta\widehat{H}_{s_i},\Delta X_{s_i}|\cF_{s_i-})}{\Var(\Delta X_{s_i}|\cF_{s_i-})},
\end{equation}
for all $i=1,\ldots,M$, as follows from \eqref{eq:LRM_strategy} using the predictability of the process $A$.
We also remark that the associated cost process $C(\varphi^H)$ is generated by the residuals of the regressions \eqref{eq:regr}.

\subsection{An example}\label{sec:LRM_example}

In this section, we illustrate the hedging approach described in Section \ref{sec:LRM} in the case of a forward-looking caplet using an RFR future as hedging instrument. This choice is motivated by the fact that, at the time of writing, SOFR futures represent the most liquidly traded products written on SOFR, while caps/floors are less liquid in the market.

We consider the model of Section \ref{sec:Hull_White}, with $Q$ playing now the role of the physical probability measure:
\begin{equation}	\label{eq:drho_P}
d\rho_t = (\alpha(t)+\beta\rho_t)dt + \sigma dW_t + dJ_t,
\end{equation}
where $J$ is defined as in \eqref{eq:process_J}, where the random variables $\{\xi_i\colon i=1,\ldots,M\}$ are independent and independent of $W$, with distribution $\cN(m_i,\gamma^2_i)$ under $Q$, for each $i=1,\ldots,M$.
For simplicity of presentation, in this subsection we assume that $\eta(dt)=dt$ (i.e., there are no roll-over dates).

As traded security, we consider a futures contract with reference period $[S,T]$, for some $S<T$. We denote by $f(t,S,T)$ the corresponding futures rate at date $t$, for $t\in[0,S]$, and define
\[
B(t,S,T) := \frac{B(T-t)-B(S-t)}{T-S}
= \frac{e^{\beta(T-t)}-e^{\beta(S-t)}}{(T-S)\beta},
\qquad \text{ for all }t\in[0,S].
\] 
We assume that the futures rate $f(\cdot,S,T)$ satisfies the following dynamics under $Q$:
\begin{equation}	\label{eq:future_dyn}
df(t,S,T) = h(t)dt + B(t,S,T)\sigma dW_t + B(t,S,T)d\widetilde{J}_t,
\end{equation}
where $\widetilde{J}$ denotes the compensated jump process defined as $\widetilde{J}_t:=J_t-\sum_{i=1}^m\ind{s_i\leq t}m_i$, for all $t\in[0,T]$, and $h:[0,T]\rightarrow\R$ is a bounded deterministic function.
The local martingale part $M$ of the discounted futures price process can be written as in \eqref{eq:MR}, with
\begin{equation}	\label{eq:future_vol}
\eta_t = (S^0_t)^{-1}B(t,S,T)\sigma
\qquad\text{and}\qquad
\Delta M_{s_i} = (S^0_{s_i})^{-1}B(s_i,S,T)(\xi_i-m_i),
\end{equation}
for all $t\in[0,S]$ and $i=1,\ldots,M$.

In the present setting, $\widehat{Z}:=\ccE(-\int_0^{\cdot}h(u)/(\sigma B(u,S,T))dW_u)$ is a square-integrable strictly positive martingale and, therefore, the minimal martingale measure $\widehat{Q}$ is given by $d\widehat{Q}=\widehat{Z}_TdQ$. 
By Girsanov's theorem, the process $\widehat{W}=(\widehat{W}_t)_{t\in[0,T]}$ defined by $\widehat{W}_t=W_t+\int_0^th(u)/(\sigma B(u,S,T))du$, for all $t\in[0,T]$, is a Brownian motion under $\widehat{Q}$.
Note that in the context of the present example the change of measure from $Q$ to $\widehat{Q}$ leaves invariant all the properties of the random variables $\{\xi_i\colon i=1,\ldots,M\}$.

\begin{remark}
In the context of the model of Section \ref{sec:Hull_White}, the futures rate $f(t,S,T)$ can be explicitly computed. Suppose that, in line with the market convention for 1-month RFR futures contracts, the futures contract settles at date $T$ at a rate quoted as $(R_T-R_S)/(T-S)$. By risk-neutral valuation under the minimal martingale measure $\widehat{Q}$, it holds that
\[
f(t,S,T) = \frac{\widehat{E}[R_T-R_S|\cF_t]}{T-S},
\qquad\text{ for all }t\in[0,S].
\]
Similarly as in Section \ref{sec:Hull_White}, under the minimal martingale measure $\widehat{Q}$ it holds that
\[
R_T-R_t 
= \rho_tB(T-t) + \widehat{A}(t,T) + \sigma\int_t^TB(T-s)d\widehat{W}_s + \sum_{i=1}^M\ind{s_i\in(t,T]}B(T-s_i)\xi_i,
\]
where $\widehat{A}(t,T) := \int_t^T\hat{\alpha}(s)B(T-s)ds$, for all $t\in[0,T]$, with $\hat{\alpha}(t):=\alpha(t)-h(t)/B(t,S,T)$ denoting the deterministic drift term in the dynamics of $\rho$ under $\widehat{Q}$.
The futures rate $f(t,S,T)$ admits then the following  representation:
\[
f(t,S,T) 
= \rho_tB(t,S,T)+\frac{\widehat{A}(t,T)-\widehat{A}(t,S)}{T-S} +\sum_{i=1}^M\ind{s_i\in(t,S]}B(s_i,S,T)m_i +\sum_{i=1}^M\ind{s_i\in(S,T]}\frac{B(T-s_i)}{T-S}m_i.
\]
\end{remark}

We suppose that the payoff $H$ to be hedged corresponds to an RFR caplet with discounted payoff
\[
H := (T-S)\bigl(F(S,T)-K\bigr)^+/S^0_T,
\]
for some $K>0$, as considered in Section \ref{sec:gaussian}. 
To determine the locally risk-minimizing strategy, we first need to compute the price process $\widehat{H}=(\widehat{H}_t)_{t\in[0,T]}$ of the payoff $H$ under the measure $\widehat{Q}$. 
This can be achieved by a direct application of Proposition \ref{prop:caplet_price}, leading to
\[
\widehat{H}_t = G(\rho_t,t,S,T,K)/S^0_t,
\]
where the function $G(\rho_t,t,S,T,K)$ is explicitly given in Proposition \ref{prop:caplet_price}, replacing $A(t,S)$ by $\widehat{A}(t,S)$ in the definition of the quantity $\Xi(S,T)$ and $\alpha(s)$ by $\hat{\alpha}(s)$ in the definition of $\Gamma_1(t,S)$.

In view of Theorem \ref{prop:LRM}, the component $\zeta^H$ of the locally risk-minimizing strategy $\varphi^H$ is determined by two terms: a first term representing the continuous Delta-hedging strategy and an additional term that takes into account the expected jump dates $\cS=\{s_1,\ldots,s_M\}$.

\begin{proposition}\label{prop:LRM_caplet}
Suppose that Assumption \ref{ass:MRP_simple} holds.
Consider a caplet delivering at date $T$ the payoff $(T-S)(F(S,T)-K)^+$, for $S\in[0,T]$ and $K>0$. 
The locally risk-minimizing strategy $\varphi^H=(\zeta^H,v^H)$ is determined by the process $\zeta^H=(\zeta^H_t)_{t\in[0,T]}$ defined by
\[
\zeta^H_t
= \zeta^{H,{\rm c}}_t\delta_{\cS^c}(t) + \zeta^{H,{\rm d}}_t\delta_{\cS}(t),
\qquad\text{ for all $t\in[0,T]$ }
\]
where, for all $i=1,\ldots,M$,
\begin{align}
\zeta^{H,{\rm c}}_t
&= \frac{-G(\rho_t,t,S,T,K)B(T-t)+B(T-S)e^{\beta(S-t)}P(t,S)\Phi(d_1(\rho_t))}{B(t,S,T)},\label{eq:LRM_caplet_c}\\
\zeta^{H,{\rm d}}_{s_i}
&= \frac{E[G(y+\xi_i,s_i,S,T,K)(\xi_i-m_i)]|_{y=\rho_{s_i-}}}{B(s_i,S,T)\gamma^2_i}.
\label{eq:LRM_caplet_d}
\end{align}
\end{proposition}
\begin{proof}
For brevity of notation, let us denote $G(x,t):=G(x,t,S,T,K)$, for all $(x,t)\in\R\times[0,S]$.
As follows from the proof of Theorem \ref{prop:LRM}, the first term on the right-hand side of \eqref{eq:LRM_strategy}, corresponding to $\zeta^{H,{\rm c}}$, is determined by the diffusive part of the process $\widehat{H}$. To this effect, we compute
\begin{align*}
\frac{\partial G(x,t)}{\partial x}
&= -B(S-t)G(x,t)+B(T-S)e^{\beta(S-t)}\bigl(e^{-xB(S-t)-\Xi(t,S)}\Phi(d_1(x))-G(x,t)\bigr)	\\
&= -G(x,t)B(T-t)+B(T-S)e^{\beta(S-t)}e^{-xB(S-t)-\Xi(t,S)}\Phi(d_1(x)).
\end{align*}
In view of equations \eqref{eq:drho_P} and \eqref{eq:future_vol}, it follows that the first component $\zeta^{H,{\rm c}}_t$ of the locally risk-minimizing strategy is given by \eqref{eq:LRM_caplet_c}.
To compute the second term on the right-hand side of equation \eqref{eq:LRM_strategy}, corresponding to $\zeta^{H,{\rm d}}_t$, observe that, in view of equation \eqref{eq:future_vol},
\[
\zeta^{H,{\rm d}}_{s_i}
= \frac{E[\Delta\widehat{H}_{s_i}\Delta M_{s_i}|\cF_{s_i-}]}{E[(\Delta M_{s_i})^2|\cF_{s_i-}]}
= \frac{E[\widehat{H}_{s_i}\Delta M_{s_i}|\cF_{s_i-}]}{E[(\Delta M_{s_i})^2|\cF_{s_i-}]}
= \frac{E[G(\rho_{s_i},s_i,S,T,K)(\xi_i-m_i)|\cF_{s_i-}]}{B(s_i,S,T)\Var(\xi_i|\cF_{s_i-})},
\]
for $i=1,\ldots,M$. Due to the independence of the random variables $\{\xi_i\colon i=1,\ldots,M\}$, it holds that
\[
\zeta^{H,{\rm d}}_{s_i}
= \frac{E[G(\rho_{s_i-}+\xi_i,s_i,S,T,K)(\xi_i-m_i)]}{(S^0_{s_i})^{-1}B(s_i,S,T)\gamma^2_i},
\qquad\text{ for all }i=1,\ldots,M,
\]
from which \eqref{eq:LRM_caplet_d} follows due to the independence of the random variables $\{\xi_i,i=1,\ldots,M\}$ from the Brownian motion $W$.
Finally, it remains to verify that $\sup_{t\in[0,S]}\widehat{H}_t\in L^2(Q)$ and $\zeta^L\in\Theta$. The first property can be shown to hold since $\widehat{H}_t\leq P(t,S)$ for all $t\in[0,S]$ and by means of standard estimates together with an application of Doob's maximal inequality. The fact that $\zeta^L\in\Theta$ follows by noting that the integral $E[\int_0^S(\zeta^H_u)^2d\langle M\rangle_u]$ can be reduced to the integration of continuous functions on the compact domain $[0,S]$ and that the function $h$ in \eqref{eq:future_dyn} is assumed to be bounded.
\end{proof}

\begin{remark}
In the context of a Vasi\v cek-type model, \cite{rutkowski2021pricing} derive an explicit replication strategy for a SOFR caplet based on SOFR futures. This is possible since their model is driven by a single source of randomness represented by a standard Brownian motion. In contrast, in our setting the presence of jumps at predetermined dates does not allow for perfect replication, thereby justifying the use of local risk-minimization.
\end{remark}

\begin{appendix}

\section{Technical proofs}

\begin{proof}[Proof of Proposition \ref{prop:Y}]
We start by computing the semimartingale characteristics $(B^Y,C^Y,\nu^Y)$ of the joint process $Y=(X,R)$. First, denoting by $B^{Y,c}$ the continuous part of the first characteristic $B^Y$, it holds that
\[
B^{Y,c}_t = \begin{pmatrix} B^{X,c}_t \\ \int_0^t\rho_sds\end{pmatrix}
= \int_0^t\biggl(\beta^Y_0(s) + \sum_{i=1}^{d+1}Y^i_{s-}\beta^Y_i(s)\biggr)ds,
\]
where $\beta^Y_0(s):=(\beta^X_0(s),\ell(s))$, $\beta^Y_i(s):=(\beta^X_i(s),\Lambda_i)$, for all $i=1,\ldots,d$, and $\beta^Y_{d+1}(s):=0$.
For the second characteristic $C^Y$, we have that 
\[
C^Y_t = \begin{pmatrix}C^X_t & 0 \\ 0 & 0\end{pmatrix}
=  \int_0^t\biggl(\alpha^Y_0(s)+ \sum_{i=1}^{d+1}Y^i_{s-}\alpha^Y_i(s)\biggr)ds,
\]
where
\[
\alpha^Y_i(s) := \begin{pmatrix}\alpha^X_i(s) & 0 \\ 0 & 0\end{pmatrix},
\text{ for all $i=0,1,\ldots,d$},
\quad\text{and}\quad
\alpha^Y_{d+1}(s):=0.
\]
The compensator $\nu^Y(dt,dx,dr)$ of the jump measure of the joint process $Y=(X,R)$ satisfies
\[
\nu^{Y,c}(dt,dx,dr) 
= \nu^{X,c}(dt,dx)\delta_{0}(dr) 
= \biggl(\mu^Y_0(t,dx,dr)+\sum_{i=1}^{d+1}Y^i_{t-}\mu^Y_i(t,dx,dr)\biggr)dt
\]
where
\[
\mu^Y_i(t,dx,dr) = \mu^X_i(t,dx)\delta_{0}(dr),
\text{ for all $i=0,1,\ldots,d$}, 
\quad\text{and}\quad
\mu^Y_{d+1}(t,dx,dr) = 0.
\]
Moreover, for all $(t,u,v)\in\R_+\times\mathcal{U}\times{\rm i}\R$,  \cite[Proposition II.1.17]{JacodShiryaev} together with \eqref{eq:SOFR_affine} and the fact that $X$ is an affine semimartingale implies that
\begin{align*}
&\delta_{\cT}(t)\int_{D\times\R}\bigl(e^{\langle u,x\rangle + vr}-1\bigr)\nu^Y(\{t\},dx,dr)	\\
&\quad= \delta_{\cT}(t)E\bigl[e^{\langle u,\Delta X_t\rangle+v\rho_t}-1|\cF_{t-}\bigr]	\\
&\quad= \delta_{\cT}(t)\left(e^{v(\ell(t)+\langle\Lambda,X_{t-}\rangle)}E[e^{\langle u+v\Lambda,\Delta X_t\rangle}-1|\cF_{t-}]+e^{v(\ell(t)+\langle\Lambda,X_{t-}\rangle)}-1\right)	\\
&\quad= \delta_{\cT}(t)\left(e^{v(\ell(t)+\langle\Lambda,X_{t-}\rangle)}\int_D(e^{\langle u+v\Lambda,x\rangle}-1)\nu^X(\{t\},dx)+e^{v(\ell(t)+\langle\Lambda,X_{t-}\rangle)}-1\right)	\\
&\quad= \delta_{\cT}(t)\left(e^{v\ell(t)+\gamma^X_0(t,u+v\Lambda)+\sum_{i=1}^dX^i_{t-}(v\Lambda_i+\gamma^X_i(t,u+v\Lambda))}-1\right).
\end{align*}
In turn, this leads to
\begin{align*}
&\int_{D\times\R}\bigl(e^{\langle u,x\rangle + vr}-1\bigr)\nu^Y(\{t\},dx,dr) 	\\
&= \delta_{\cT^c}(t)\int_{D}\bigl(e^{\langle u,x\rangle}-1\bigr)\nu^X(\{t\},dx)
+\delta_{\cT}(t)\int_{D\times\R}\bigl(e^{\langle u,x\rangle+ vr}-1\bigr)\nu^Y(\{t\},dx,dr) \\
&= \delta_{\cT^c}(t)\left(e^{\gamma^X_0(t,u)+\sum_{i=1}^dX^i_{t-}\gamma^X_i(t,u)}-1\right) 
+ \delta_{\cT}(t)\left(e^{v\ell(t)+\gamma^X_0(t,u+v\Lambda)+\sum_{i=1}^dX^i_{t-}(v\Lambda_i+\gamma^X_i(t,u+v\Lambda))}-1\right)	\\
&= e^{\gamma^Y_0(t,u,v)+\sum_{i=1}^{d+1}Y^i_{t-}\gamma^Y_i(t,u,v)}-1, 
\end{align*}
where
\begin{equation}		\label{eq:gammaY}	\begin{aligned}
\gamma_0^Y(t,u,v) &:= \delta_{\cT^c}(t)\gamma^X_0(t,u) 
+ \delta_{\cT}(t)\bigl(v\ell(t)+\gamma^X_0(t,u+\Lambda v)\bigr),	\\
\gamma_i^Y(t,u,v) &:= \delta_{\cT^c}(t)\gamma^X_i(t,u) 
+ \delta_{\cT}(t)\bigl(v\Lambda_i+\gamma^X_i(t,u+\Lambda v)\bigr),
\quad\text{ for all }i=1,\ldots,d,	\\
\gamma_{d+1}^Y(t,u,v) &:= 0.
\end{aligned}	\end{equation}
In particular, note that $\gamma^Y_i(t,u,v)=0$ for all $(t,u,v)\in(\R_+\setminus(\cT\cup J^X))\times\mathcal{U}\times{\rm i}\R$ and $i=0,1,\ldots,d+1$. It follows that the parameter set $(A^Y,\beta^Y,\alpha^Y,\mu^Y,\gamma^Y)$ is {\em good} in the sense of \cite[Definition 3.1]{keller-ressel2019}. 
Moreover, since the affine semimartingale $X$ is assumed to be infinitely divisible, \cite[Lemma 4.4]{keller-ressel2019} implies that, for all $t\in J^X$ and $i=0,1,\ldots,d$,
\[
\gamma^X_i(t,u) = \langle\tilde{\beta}^X_i(t),u\rangle + \frac{1}{2}\langle u,\tilde{\alpha}^X_i(t)u\rangle + \int_{D\setminus\{0\}}\bigl(e^{\langle x,u\rangle}-1-\langle h(x),u\rangle\bigr)\tilde{\mu}^X_i(t,dx),
\qquad\text{for all }u\in\mathcal{U},
\]
for suitable $\tilde{\beta}_i^X(t)\in\R^d$, $\tilde{\alpha}^X_i(t)\in\mathcal{S}^d$ and Borel measures $\tilde{\mu}^X_i(t,\cdot)$ on $D\setminus\{0\}$.
Making use of the notation $w=(u,v)\in D\times\R$ and $y=(x,r)$ and in view of \eqref{eq:gammaY}, this implies that 
\[
\gamma^Y_i(t,w) = \langle\tilde{\beta}^Y_i(t),w\rangle + \frac{1}{2}\langle w,\tilde{\alpha}^Y_i(t)w\rangle + \int_{(D\setminus\{0\})\times\R}(e^{\langle y,w\rangle}-1-\langle \tilde{h}(y),w\rangle)\tilde{\mu}^Y_i(t,dy),
\qquad\text{for all }w\in\mathcal{U}\times{\rm i}\R,
\]
for all $t\in \cT\cup J^X$ and $i=0,1,\ldots,d+1$, where we set
\begin{align*}
\tilde{\beta}^Y_0(t) &:= \begin{pmatrix} \tilde{\beta}^X_0(t) \\ (\ell(t)+\langle\tilde{\beta}^X_0(t),\Lambda\rangle+\int_D(\tilde{h}_{d+1}(\langle\Lambda,x\rangle)-\langle\Lambda,h(x)\rangle)\tilde{\mu}^X_0(t,dx))\delta_{\cT}(t)\end{pmatrix},	\\[1mm]
\tilde{\beta}^Y_i(t) &:= \begin{pmatrix} \tilde{\beta}^X_i(t) \\ (\Lambda_i+\langle\tilde{\beta}^X_i(t),\Lambda\rangle+\int_D(\tilde{h}_{d+1}(\langle\Lambda,x\rangle)-\langle\Lambda,h(x)\rangle)\tilde{\mu}^X_i(t,dx))\delta_{\cT}(t)\end{pmatrix},
\text{ for all }i=1,\ldots,d, \\[1mm]
\tilde{\alpha}^Y_i(t) &:= \begin{pmatrix}\tilde{\alpha}^X_i(t) & \tilde{\alpha}^X_i(t)\Lambda\delta_{\cT}(t) \\ \Lambda^{\top}\tilde{\alpha}^X_i(t)\delta_{\cT}(t) & \Lambda^{\top}\tilde{\alpha}^X_i(t)\Lambda\delta_{\cT}(t)\end{pmatrix}\in\mathcal{S}^{d+1},
\text{ for all }i=0,1,\ldots,d,\\
\tilde{\mu}^Y_i(t,dy) &=\tilde{\mu}^Y_i(t,dx,dr):=\tilde{\mu}^X_i(t,dx)\bigl(\delta_{\langle\Lambda,x\rangle}(dr)\delta_{\cT}(t)+\delta_0(dr)\delta_{\cT^c}(t)\bigr),
\end{align*}
with $\tilde{h}:\R^{d+1}\rightarrow\R^{d+1}$ being a truncation  function satisfying $\tilde{h}_i(y)=h_i(x)$, for all $i=1,\ldots,d$.

Moreover, we set $\tilde{\beta}^Y_{d+1}(t):=0$, $\tilde{\alpha}^Y_{d+1}(t):=0$ and $\tilde{\mu}^Y_{d+1}(t,dy):=0$ for all $t\in\R_+$. 
For all $i=0,1,\ldots,d$, the measure $\tilde{\mu}^Y_i(t,dy)$ is a L\'evy measure on $(D\setminus\{0\})\times\R$. This follows by observing that, as a consequence of Cauchy-Schwarz inequality,
\[
\int_{(D\setminus\{0\})\times\R}(1\wedge\|y\|^2)\tilde{\mu}^Y_i(t,dy)
\leq (1+\|\Lambda\|^2)\int_{D\setminus\{0\}}(1\wedge\|x\|^2)\tilde{\mu}^X_i(t,dx) <\infty,
\]
for all $t\in\cT$ and $i=0,1,\ldots,d$.
Since the affine semimartingale $X$ satisfies by assumption the conditions of \cite[Proposition 5.2]{keller-ressel2019}, its associated enhanced parameter set is admissible, in the sense of \cite[Definition 5.1]{keller-ressel2019}. In turn, this implies that the enhanced parameter set of $Y$, determined by $(\tilde{\beta}^Y,\tilde{\alpha}^Y,\tilde{\mu}^Y)$ as defined above, is also admissible.
Therefore, by \cite[Theorem 5.7]{keller-ressel2019}, on the canonical stochastic basis $(\Omega',\cF',(\cF'_t)_{t\geq0},\P')$ there exists an infinitely divisible Markov process $Y'=(Y'_t)_{t\geq0}$ with $Y'_0=(x,0)$ that is an affine semimartingale with characteristics $(B^Y,C^Y,\nu^Y)$ as computed above.
Since the two Markov processes $Y=(X,R)$ and $Y'$ have the same characteristics and the process $Y'$ is unique in law, it follows that $Y=(X,R)$ and $Y'$ have the same law (compare with \cite[Lemmata 10.1 and 10.2]{DuffieFilipovicSchachermayer}). 
Denoting by $E'$ the expectation under the measure $P'$, this implies that
\begin{equation}	\label{eq:CF_Y_proof}
E[e^{\langle w,Y_T\rangle}] 
= E'[e^{\langle w,Y'_T\rangle}] 
= e^{\Phi_0(T,w) + \langle\Psi_0(T,w),x\rangle},
\end{equation}
for all $w=(u,v)\in\mathcal{U}\times{\rm i}\R$ and $0\leq t\leq T<\infty$, where the functions $\Phi_0(T,w)$ and $\Psi_0(T,w)$ are solution to  \eqref{eq:RicattiY_1}-\eqref{eq:RicattiY_4}, as follows from \cite[Theorem 3.1]{keller-ressel2019} together with the specific structure of the characteristics $(B^Y,C^Y,\nu^Y)$ computed in the first part of the proof. The conditional version of the Fourier transform \eqref{eq:CF_Y} follows from \eqref{eq:CF_Y_proof} by relying on the Markov property of $Y$ on the stochastic basis $(\Omega,\cF,(\cF_t)_{t\geq0},P)$.
\end{proof}

\begin{prooflemma}{\ref{lemma:A1}}
	Firstly, taking expectation of equation \eqref{eq:335} immediately yields $m(t, T)$ and thus \eqref{mean1}. 
	Regarding the covariance, for the continuous part we note that
	\begin{align*}
		E\bigg[\int_t^{T_1} e^{\beta(T_1-u)} dW_u \cdot \int_t^{T_2} e^{\beta(T_2-v)}dW_v\bigg] = 
		\int_t^{T_1 \wedge T_2} e^{\beta(T_1 + T_2 - 2 u)} du = \frac{e^{\beta(T_1+T_2)}}{2 \beta}\Big( e^{-2\beta t} - e^{- 2 \beta(T_1 \wedge T_2)} \Big). 
	\end{align*}
	 Next, we compute the conditional covariance of the jumps:
	\begin{align*}
		\Cov \Big( & \sum_{i=1}^M \ind{s_i \in (t,T_1] } e^{ \beta (T_1-s_i)} \xi_i, \sum_{i=1}^M \ind{s_i \in (t, T_2]} e^{ \beta (T_2-s_i)} \xi_i \, \Big| \, \rho_t  \Big) 
		=\sum_{i=1}^M \ind{s_i \in (t,T_1 \wedge T_2] } e^{ \beta (T_1 + T_2-2s_i)} \Var(\xi_i).
	\end{align*}
	Putting the two parts together we obtain \eqref{mean1}.
\end{prooflemma}

\begin{proof}[Proof of Lemma \ref{lemma:A2,3}] By Fubini's theorem, we have that
\begin{align*}
E[R_T | \rho_t,R_t] &=R_t + \int_{(t,T]} m(t,u)  \eta(du) = R_t + \int_{(t,T]} m(t,u)  du + \sum_{j=1}^{N} \ind{t_j \in (t,T]} m(t, t_j).
\end{align*}
In view of \eqref{mean1}, the first integral on the right-hand side can be computed as follows:

	\begin{align*}
		\int_{(t,T]} m(t,u)  du & = \int_t^T \Big( \rho_t e^{\beta (u-t)} + a(t,u)  +\sum_{i=1}^M \ind{s_i \in (t,u]} m_i e^{\beta(u-s_i)} \Big) du \\
		& = \rho_t B(T-t) + A(t,T) + \sum_{i=1}^M \ind{s_i \in (t,T]} m_i B(T-s_i),
	\end{align*}
	while the second term is given by 
	\begin{align*}
		\sum_{j=1}^N \ind{t_j \in (t,T]} m(t, t_j) & = \sum_{j=1}^N \ind{t_j \in (t,T]} \Big( \rho_{t} e^{\beta(t_j-t)} + a(t,t_j) +
		\sum_{i=1}^M \ind{s_i \in (t,t_j]} m_i e^{\beta(t_j-s_i)} \Big). 
	\end{align*}
Regarding the conditional variance, we observe that by Fubini's theorem it holds that
\begin{align}
&\Var(R_T | \rho_t, R_t ) = \dinttT c(t,u,v) \eta(dv) \eta(du) \notag\\
&\quad =  \dinttT\left( \frac{\sigma^2e^{\beta(u+v)}}{2 \beta}\Big( e^{-2\beta t} - e^{- 2 \beta(u \wedge v)} \Big) + \sum_{i=1}^M\ind{s_i \in  (t,u \wedge v]} \gamma_i^2 e^{ \beta (u + v-2s_i)}\right) \eta(dv) \eta(du).
\label{eq:var0}
\end{align}	
We first compute
\[
\frac{\sigma^2e^{-2\beta t}}{2 \beta} \int_{(t,T]} \int_{(t,T]} e^{\beta(u+v)} \eta(dv) \eta(du) 
= \frac{\sigma^2e^{-2\beta t}}{2 \beta} \bigg( \int_{(t,T]}  e^{\beta u} \eta(du)\bigg)^2 
= \frac{\sigma^2}{2 \beta} B'(t,T)^2.
\]
Then, using again Fubini's theorem, we can compute
\begin{align}
\dinttT e^{\beta(u+v-2(u\wedge v))}\eta(dv)\eta(du)	
&= \int_{(t,T]}\int_{(t,u]}e^{\beta(u-v)}\eta(dv)\eta(du)
+ \int_{(t,T]}\int_{(u,T]}e^{\beta(v-u)}\eta(dv)\eta(du)	\notag\\
&= 2 \int_{(t,T]}\int_{(t,u)}e^{\beta(u-v)}\eta(dv)\eta(du)
+ \int_{(t,T]}\eta(\{u\})\eta(du),
\label{eq:var2}
\end{align}
Focusing on the first integral in \eqref{eq:var2}, we compute
\begin{align*}
\int_{(t,T]}\int_{(t,u)}e^{\beta(u-v)}\eta(dv)\eta(du)
&= \int_{(t,T]}\int_{(t,u)}e^{\beta(u-v)}dv\eta(du)
+ \sum_{j=1}^N\int_{(t,T]}\ind{t_j\in(t,u)}e^{\beta(u-t_j)}\eta(du)	\\
&= \frac{1}{\beta}\Bigl(B'(t,T)-(T-t)\Bigr)+\sum_{j=1}^N\ind{t_j\in(t,T]}\Bigl(B'(t_j,T)-\frac{1}{\beta}\Bigr)	.
\end{align*}
The last integral in \eqref{eq:var2} reduces to
\[
\int_{(t,T]}\eta(\{u\})\eta(du)
= \sum_{k=1}^N\int_{(t,T]}\ind{t_k=u}\eta(du)
= \sum_{j,k=1}^N\ind{t_j\in(t,T]}\ind{t_k=t_j}
= \sum_{j=1}^N\ind{t_j\in(t,T]}.
\]
Applying a reasoning analogous to \eqref{eq:var2}, we can rewrite as follows the second term in equation \eqref{eq:var0}, omitting to write the term $\gamma_i^2$ for simplicity of presentation:
\begin{align}
&\dinttT \sum_{i=1}^M\ind{s_i \in  (t,u \wedge v]}  e^{ \beta (u + v-2s_i)} \eta(dv) \eta(du)	\notag\\
&\quad= \int_{(t,T]}\int_{(t,u]} \sum_{i=1}^M\ind{s_i \in  (t,v]}  e^{ \beta (u + v-2s_i)} \eta(dv) \eta(du)
+ \int_{(t,T]}\int_{(u,T]} \sum_{i=1}^M\ind{s_i \in  (t,u]} e^{ \beta (u + v-2s_i)} \eta(dv) \eta(du)	\notag\\
&\quad = 2\int_{(t,T]}\int_{(t,u)} \sum_{i=1}^M\ind{s_i \in  (t, v]}  e^{ \beta (u + v-2s_i)} \eta(dv) \eta(du)
 + \int_{(t,T]} \sum_{i=1}^M\ind{s_i \in  (t,u]} e^{ 2\beta (u -s_i)} \eta(\{u\}) \eta(du). \notag
\end{align}
The first integral appearing in the last line can be computed as follows
\begin{align*}
&	\int_{(t,T]}\int_{(t,u)} \sum_{i=1}^M\ind{s_i \in  (t, v]} e^{ \beta (u + v-2s_i)} \eta(dv) \eta(du)	\\
&\quad = \sum_{i=1}^M\ind{s_i\in(t,T]}\frac{\bar{B}(s_i,T,2\beta)-B'(s_i,T)}{\beta}
+ \sum_{i=1}^M\sum_{j=1}^N\ind{s_i\in(t,T]}\ind{t_j\in[s_i,T]}e^{2\beta(t_j-s_i)}I(t_j,T),
\end{align*}
while the second integral reduces to
\begin{align*}
\int_{(t,T]} & \sum_{i=1}^M\ind{s_i \in  (t,u]} e^{ 2\beta (u -s_i)} \eta(\{u\}) \eta(du)
= \sum_{k=1}^N\sum_{i=1}^M\int_{(t,T]}\ind{s_i\in(t,u]}\ind{t_k=u}e^{ 2\beta (u -s_i)}\eta(du)	\\
&= \sum_{j,k=1}^N\sum_{i=1}^M\ind{t_j\in(t,T]}\ind{s_i\in(t,t_j]}\ind{t_k=t_j}e^{ 2\beta (t_j-s_i)}	
= \sum_{i=1}^M\sum_{j=1}^N\ind{s_i\in(t,T]}\ind{t_j\in[s_i,T]}e^{ 2\beta (t_j-s_i)}.
\end{align*}
\end{proof}

\end{appendix}

 \bibliographystyle{agsm}
 \bibliography{tsloc.bib}

\end{document}